\documentclass[cmp]{svjour}  

\usepackage[usenames,dvipsnames]{xcolor}
\usepackage[pdftex]{graphicx}  
\usepackage{epstopdf} 
\DeclareGraphicsExtensions{.pdf,.eps,.png,.jpg,.mps} 
\usepackage{epsfig} 
\usepackage[caption=false]{subfig}
\usepackage{float}
\usepackage{wrapfig}

\journalname{Communications in Mathematical Physics}

\usepackage{amsmath,amssymb,bbm,bm}
\usepackage{mathtools}
\usepackage{centernot}  
\usepackage{wasysym}  
\usepackage{relsize}       

\usepackage{array}    
\usepackage{fancyhdr}
\usepackage{import}  
\usepackage{setspace}
\usepackage{exscale}
\usepackage[english]{babel}

\newcommand{\mc}{\mathcal}            
\newcommand{\matr}[1]{\mathsf{#1}} 

\newcommand{\G}{\mathcal{G}}   

\newcommand{\ot}{\leftarrow}    

\newcommand{\up}{\uparrow}                                
\newcommand{\down}{\downarrow}                       

\begin{document}

\title{An Explicit Bound for Dynamical Localisation in an Interacting Many-Body System}
\titlerunning{Dynamical Localisation in an Interacting Many-Body System}

\author{P.-L. Giscard\inst{1}  \and Z. Choo\inst{2} \and M. T. Mitchison\inst{3}\fnmsep\inst{1} \and J. J. Mendoza-Arenas\inst{1} \and D. Jaksch\inst{1}\fnmsep\inst{4}}
\institute{University of Oxford, Department of Physics, Clarendon Laboratory, Oxford OX1 3PU, UK, \and University of Oxford, Department of Statistics, 1 South Parks Road, Oxford OX1 3TG, UK \and Quantum Optics and Laser Science Group, Blackett Laboratory, Imperial College London, London SW7 2BW, UK, \and Centre for Quantum Technologies, National University of Singapore, 3 Science Drive 2, Singapore 117543,\\ \email{p.giscard1@physics.ox.ac.uk}}
\authorrunning{P.-L. Giscard, Z. Choo, M. Mitchison, J. J. Mendoza-Arenas, D. Jaksch}



\date{\today}

\maketitle
\begin{abstract}
We characterise and study dynamical localisation of a finite interacting quantum many-body system. We present explicit bounds on the disorder strength required for the onset of localisation of the dynamics of arbitrary ensemble of sites of the XYZ spin-1/2 model.
We obtain these results using a novel form of the fractional moment criterion, which we establish, together with a generalisation of the self-avoiding walk representation of the system Green's functions, called path-sums. These techniques are not specific to the XYZ model and hold in a much more general setting.
We further present bounds for two observable quantities in the localised regime: the magnetisation of any sublattice of the system as well as the linear magnetic response function of the system. We confirm our results through numerical simulations.
\end{abstract}

\keywords{Dynamical localisation, Many-body localisation, Interactions, Path-sums, Fractional moments}

\section{Introduction}
The mathematical and physical understanding of Anderson localisation for a single particle in a lattice with random on-site potentials has greatly progressed since it was first recognised in 1958 \cite{Anderson_PRB_1958,Kramer_RPP_1993,Hundertmark2007}. In contrast, the theory of many-body localisation is still in its infancy, despite the large amount of recent research focusing on localisation in disordered many-particle systems (for example Refs. \cite{Gornyi_PRL_2005,Basko_AoP_2006,Oganesyan_PRB_2007,Znidaric_PRB_2008,karahaliosPRB2009,berkelbach2010prb,barisicPRB2010,Monthus_PRB_2010,Aleiner_NatPhys_2010,Pal2010,Khatami2012,albrecht2012prb}).  In particular, there is still no clear consensus on the precise definition of the many-body localised phase, with a number of different definitions having been proposed in the recent literature (e.g.\ Refs. \cite{Basko_AoP_2006,Znidaric_PRB_2008,Aizenman2009,Monthus_PRB_2010,albrecht2012prb,Bauer_arXiv_2013}). Consequently, a global picture of many-body localisation is still lacking.



An important advance towards the conceptual understanding of the problem was the introduction of the notion of localisation in the quantum configuration or Fock space \cite{Altshuler_PRL_1997}, which has emerged as a powerful framework for understanding many-body localisation \cite{Gornyi_PRL_2005,Basko_AoP_2006,Monthus_PRB_2010,Bauer_arXiv_2013}. Configuration-space localisation has also proved useful in understanding the absence of thermalisation in integrable systems \cite{Canovi2011} and glassy dynamics in strongly interacting systems \cite{Carleo2012}. 
We treat the evolution of the many-body quantum state as that of a single fictitious particle moving in an exponentially large configuration space. This viewpoint leads us quite naturally to define many-body dynamical localisation as Anderson localisation of this effective one-particle problem. Furthermore, it facilitates a non-perturbative, mathematically rigorous treatment.

Rigorous mathematical results concerning localisation in interacting many-body systems  have been established in e.g. \cite{Aizenman2009} and \cite{Chulaevsky2009}. However, the critical disorder inducing a many-body localisation transition as predicted by these studies is difficult to ascertain and the physical consequences of the localisation remain, in part, unexplored. 
The contribution of the present work to the study of many-body localisation is thus three-fold. 
First we provide an explicit bound for the critical amount of disorder necessary for localisation to occur. 
Second, we make concrete predictions of experimental signatures of the localisation by providing bounds for two physically observable quantities in the localised regime: the magnetisation of any finite ensemble of sites immersed in the system as well as  the linear magnetic response function. We confirm our predictions with Time Evolving Block Decimation (TEBD) simulations of a 1D system.
Third, our approach to many-body localisation differs from existing approaches \cite{Aizenman2009} and \cite{Chulaevsky2009}. While we do use a form of fractional moment criterion as a mathematical signature of localisation, we rely on a novel representation of the Green's functions associated to any sublattice of the full system to prove that this criterion is met for a finite amount of disorder. This representation, known as path-sums \cite{Giscard2012}, is a generalisation of the self-avoiding walk representation exploited in the one-body setting \cite{Hundertmark2007}. The path-sums approach 
facilitates the manipulation and systematic representation of matrix resolvents and Green's functions. We hope that it will ultimately contribute to the proof of localisation in \textit{infinite} interacting many-body systems, which remains elusive.

The article is organised as follows. In \S \ref{model} we present the model system and general mathematical strategy we employ. For the sake of concreteness we focus on the particular case of a finite, interacting spin-$\frac{1}{2}$ system described by the XYZ-Hamiltonian. This section culminates in a criterion characterising the localised regime. In \S \ref{ManyBodyBound} we present a bound on the amount of disorder required for the criterion introduced in the preceding section to be fulfilled. We then provide constraints for the magnetisation of arbitrary sublattices of the system as well as the correlations between spins at sites of the sublattice in the localised regime.
The bound on the latter quantity is shown to imply a bound on the linear magnetic response function to magnetic perturbations, which we find to be vanishingly small deep in the localised regime.
We conclude in \S\ref{conclu}. 
The proofs of all technical results are given in the appendices. Appendix \ref{MBproof} contains the proof of the main result concerning the onset of localisation and a brief presentation of the path-sums approach. In Appendix \ref{TechnicRes} we gather the proofs of technical lemmas necessary to the proof presented in Appendix \ref{MBproof}.
Finally, in Appendices \ref{LocMagn} and \ref{Localcorrel} we derive the results concerning the magnetisation and correlations in the localised regime.

%
%

\section{Model system and approach}\label{model}
\subsection{Model system}
\subsubsection{Hamiltonian}
\noindent We consider a lattice of $N$ spin-1/2 particles governed by the XYZ-Hamiltonian, which represents one of the simplest models of strongly-interacting many-body systems. In one dimension, this is equivalent to a tight-binding model of spin-polarised fermions with nearest-neighbour density-density interaction. However, we consider the general case of an arbitrary number of dimensions. The Hamiltonian reads
\begin{equation}\label{HamiltonianXYZ}
\matr{H}=\sum_{i} B_i \sigma_i^z+\sum_{\langle i j \rangle}\big\{J_x\sigma_{i}^x\sigma_{j}^x+J_y\sigma_{i}^y\sigma_{j}^y+\Delta\sigma_{i}^z\sigma_{j}^z\big\}, 
\end{equation}
where $J_x$, $J_y$ and $\Delta$ parametrise the spin-spin interactions, and  $\sum_{\langle i j \rangle}$ denotes a sum over nearest neighbours. The $\{B_i\}_{i\in[1,N]}$ variables are local magnetic fields which we assume to be independent and identically distributed (iid) random variables taken from a Gaussian distribution with standard deviation $\sigma_B$. For convenience, we choose 
the basis states of the Hilbert space to be configurations of spins along the direction of the magnetic fields, e.g. for a two spin-1/2 system this is $\{|\hspace{-1mm}\up\up\rangle,\,|\hspace{-1mm}\up\down\rangle,\,|\hspace{-1mm}\down\up\rangle,\,|\hspace{-1mm}\down\down\rangle\}$, where $\up$ and $\down$ designate the situations where the spin is aligned and antialigned with the magnetic fields, respectively. 

For $J_x\neq J_y$, the XYZ-Hamiltonian does not conserve the total number of up and down spins present at any time in the system. Since we can add the sum of all the random magnetic fields $\sum_{i} B_i$ to the Hamiltonian without affecting the dynamics, an up spin is effectively equivalent to a particle seeing a random onsite potential $2B_i$ at site $i$, while a down spin always sees 0 onsite potential, being effectively equivalent to the absence of a particle. In this sense, the XYZ Hamiltonian \emph{does not conserve the number of particles} present in the system.
At the opposite, when $J_x=J_y$, the number of up- and down-spins present in the system is conserved by $\matr{H}$. We call this specific situation the XXZ regime.

\subsubsection{Mathematical description}
We consider the total system described by the XYZ Hamiltonian \eqref{HamiltonianXYZ}, denoted by $\mathbb{S}$, as comprising two parts $S$ and $S'$. We shall use a partial basis of the configuration space of $\mathbb{S}$ that specifies only the state of $S'$ \footnote{We use the notation $S$ and $S'$ for consistency with Ref. \cite{GiscardWS}, where the use of of system partitions in quantum mechanics was introduced.}. From hereon we designate the configurations available to $S$ and $S'$ by Roman and Greek letters, respectively. 

Mathematically this procedure corresponds to working with a tensor product partition of the Hamiltonian into smaller matrices \cite{Giscard2012}. 
These small matrices, denoted $\matr{H}_{\omega\ot\alpha}$, are submatrices of $\matr{H}$ of size $2^{|S|}\times 2^{|S|}$ with $|S|$ the number of sites in $S$ and fulfill the characteristic relation of tensor product partitions \cite{Giscard2012}
\begin{subequations}
\begin{equation}
|\omega\rangle\langle \alpha|\otimes \matr{H}_{\omega\ot\alpha}=\big(\matr{P}_{\omega}\otimes \matr{I}_S\big)\,.\,\matr{H}\,.\,\big(\matr{P}_{\alpha}\otimes \matr{I}_S\big),
\end{equation}
or equivalently
\begin{equation}\label{PartitionRelation}
\matr{H}_{\omega\ot\alpha}=\big(\langle \omega|\otimes \matr{I}_S\big)\,.\,\matr{H}\,.\,\big(|\alpha\rangle\otimes \matr{I}_S\big),
\end{equation}
\end{subequations}
with $\matr{I}_S$ the identity matrix on the spins of $S$ and $\matr{P}_{\alpha}=|\alpha\rangle\langle\alpha|$ the projector onto configuration $\alpha$ of $S'$. The matrix $\matr{H}_{\alpha}\equiv \matr{H}_{\alpha\ot\alpha}$ is a small effective Hamiltonian governing the evolution of $S$ when $S'$ is static in configuration $\alpha$. For this reason the $\matr{H}_{\alpha}$ matrices are called statics. Similarly, the matrix $\matr{H}_{\omega\ot\alpha}$ represents the impact on $S$ of a transition of $S'$ from configuration $\alpha$ to $\omega$ and for this reason is called a flip.
For the choice $S=\emptyset$, $S'=\mathbb{S}$ the statics and flips of the Hamiltonian identify with its matrix elements $\matr{H}_{\omega\ot\alpha}=(\matr{H})_{\omega\alpha}$. For the opposite choice  $S=\mathbb{S}$, $S'=\emptyset$, there is no flip and only one static, which is the Hamiltonian itself $\matr{H}_{\alpha}=\matr{H}$. In general, the statics and flips can be thought of as generalised matrix elements, which interpolate between the entries of $\matr{H}$ and $\matr{H}$ itself.
Thanks to Eq.~(\ref{PartitionRelation}), we find the statics of a tensor product partition of the XYZ Hamiltonian, Eq.~(\ref{HamiltonianXYZ}), to be
%
%
\begin{align}
\matr{H}_{\alpha}&=\sum_{i\in S}B_i\sigma_i^z+\sum_{i\in S}\sum_{\substack{j\in\langle i\rangle\\ j\in S} }(J_x\sigma_i^x\sigma_j^x+J_y\sigma_i^y\sigma_j^y+\Delta \sigma^z_i\sigma_j^z)\nonumber\\
&\hspace{5mm}+\sum_{i\in S}\sum_{j\in\langle i\rangle\atop j\in S'}\Delta \sigma^z_{i}(1-2\delta_{\down_j,\,\alpha})+\sum_{i\in S'}B_i(1-2\delta_{\down_i,\,\alpha})\matr{I}_S\nonumber\\
&\hspace{5mm}+\Delta\sum_{i\in S'}\sum_{\substack{j\in \langle i\rangle\\j\in S'}}(1-2\delta_{\down_i\up_{j},\,\alpha}-2\delta_{\up_i\down_{j},\,\alpha})\matr{I}_S,\label{Ha}
\end{align}
where, $|\alpha|$ is the number of up-spins at sites of $S'$ when it is in configuration $\alpha$, and, for example, $\delta_{\down_j,\,\alpha}$ is 1 if the spin at site $j$ is down in configuration $\alpha$ and 0 otherwise. For latter convenience, we define $Y_\alpha=\sum_{i\in S'}B_i(1-2\delta_{\down_i,\,\alpha})$  a random variable representing the disorder due solely to $S'$ in the static $\matr{H}_\alpha$. The static is thus conveniently separated into two pieces as
\begin{equation}\label{StaticFormY}
\matr{H}_\alpha= Y_\alpha\, \matr{I}_S + \tilde{\matr{H}}_{\alpha},
\end{equation}
where $\tilde{\matr{H}}_{\alpha}$ depends deterministically on $\alpha$ and comprises the random $B$-fields  at sites of $S$ but does not depend on any of the random $B$-fields at sites of $S'$. 
The flips of the Hamiltonian are
%
%
\begin{align}
\matr{H}_{\omega\ot\alpha}&=\begin{cases}(\matr{H}_{\omega\ot\alpha})_{ji}=J_x+J_y,\,&\text{if }(\alpha i)\text{ and }(\omega j)\text{ differ by the flips of two}\\
&\text{neighboring spins as follows $\up\down\longleftrightarrow\down\up$},\\
(\matr{H}_{\omega\ot\alpha})_{ji}=J_x-J_y,\,&\text{if }(\alpha i)\text{ and }(\omega j)\text{ differ by the flips of two}\\
&\text{neighboring spins as follows $\up\up\longleftrightarrow\down\down$}\\
0,&\text{otherwise,}
\end{cases}
\end{align}
with for example $(\alpha i)$ the configuration of $\mathbb{S}$ where $S'$ and $S$ are in configuration $\alpha$ and $i$, respectively. 
In the following we will only need the two-norm of the flips, which is bounded as follows $\|\matr{H}_{\omega\ot\alpha}\|\leq 2(J_x^2+J_y^2)^{1/2}:=2\mc{J}$.


\subsection{Dynamics of the system equivalent particle}\label{DynSEP}~\\

\noindent Our goal is to establish dynamical localisation of $S'$ regardless of the state of $S$. In many-body situations, it is useful to consider localisation of the system directly in its configuration space rather than on the physical lattice \cite{Basko_AoP_2006,Monthus_PRB_2010,Aizenman2009}. As long as the system is finite, the configuration space associated to any partition $\mathbb{S}=S\cup S'$ of the system is a weighted discrete graph $\G_{S,S'}$, which we construct as follows:
\begin{itemize}
\item[i)] For each configuration $\alpha$ available to $S'$, we draw a vertex $v_\alpha$.
\item[ii)] If $\matr{H}_{\omega\ot\alpha}\neq0$, we draw an edge $(\alpha\omega)$ from vertex $v_{\alpha}$ to vertex $v_{\omega}$. 
\item[iii)] To each edge $(\alpha\omega)$, we associate the weight $\matr{H}_{\omega\ot\alpha}$.
\end{itemize}
Finally, we define $d(\alpha,\omega)$ the distance between two configurations of the system as the length of the shortest walk on $\G_{S,S'}$ from vertex $v_{\alpha}$ to vertex $v_{\omega}$. 

The configuration graph gives an alternative picture of the system dynamics.
For the choice $S=\emptyset$, $S'=\mathbb{S}$ the statics and flips are the matrix entries of the Hamiltonian $\matr{H}_{\omega\ot\alpha}
=(\matr{H})_{\omega\alpha}$ and the dynamics in real space is equivalent to that of a single particle in the configuration space, which we call the System Equivalent 
Particle (SEP), undergoing a continuous time quantum random walk on $
\G_{\emptyset,\mathbb{S}}$. 
Furthermore, in the XXZ regime where $J_x=J_y$,  if there is initially single up-spin in a sea of down-spins (or the opposite), the configuration graph $\G_{\emptyset,\mathbb{S}}$ is simply the real-space lattice and the SEP represents this single up-spin. 

For general partitions $\mathbb{S}=S\cup S'$, $S$ can be interpreted as the internal degree of freedom of the SEP, while its trajectory on the graph $\G_{S,S'}$ represents the evolution of $S'$. When $S'$ is in configuration $\alpha$, the SEP is on site $v_\alpha$ and feels the random potential $Y_\alpha$, which we then call the configuration potential at $\alpha$. Depending on the amount of disorder present in the various configuration potentials, the SEP may be dynamically localised, which corresponds to many-body localisation of $S'$, regardless of the state of $S$.

\subsection{Criterion for many-body localisation}\label{CritMB}~\\

\noindent In analogy with the case of a single particle on a lattice \cite{Hundertmark2007}, the natural signature of SEP dynamical localisation is simply the absence of diffusion of the SEP on $\G_{S,S'}$. This is expressed rigorously using the same tensor product partition that we used for the Hamiltonian but this time for the evolution operator $\matr{P}_{I}(\matr{H}) e^{-i \matr{H} t}$, where $\matr{P}_{I}$ is the projector onto the eigensubspace of $\matr{H}$ with energy in bounded interval $I$. The elements of this partition fulfill the characteristic relation
\begin{subequations}
\begin{equation}
|\omega\rangle\langle \alpha|\otimes \Big(\matr{P}_{I}(\matr{H}) e^{-i \matr{H} t}\Big)_{\omega\ot\alpha}=\big(\matr{P}_{\omega}\otimes \matr{I}_S\big)\,.\,\matr{P}_{I}(\matr{H}) e^{-i \matr{H} t}\,.\,\big(\matr{P}_{\alpha}\otimes \matr{I}_S\big),
\end{equation}
or equivalently
\begin{equation}
\Big(\matr{P}_{I}(\matr{H}) e^{-i \matr{H} t}\Big)_{\omega\ot\alpha}=\big(\langle \omega|\otimes \matr{I}_S\big)\,.\,\matr{P}_{I}(\matr{H}) e^{-i \matr{H} t}\,.\,\big(|\alpha\rangle\otimes \matr{I}_S\big).
\end{equation}
\end{subequations}
For convenience, we say that $(\matr{P}_{I}(\matr{H}) e^{-i \matr{H} t})_{\omega\ot\alpha}$ is a flip of the evolution operator when $\omega\neq \alpha$ and, otherwise, a static of the evolution operator.
In this formalism, localisation of the SEP on the configuration graph $\G_{S,S'}$ is characterised by an exponential bound on the norm of the flips of the evolution operator ($\hbar=1$)
\begin{subequations}\label{DynLocall}
\begin{equation}\label{DynLoc1}
\mathbb{E}\bigg[\sup_{t\in\mathbb{R}}\Big\|\Big(\matr{P}_{I}(\matr{H}) e^{-i \matr{H} t}\Big)_{\omega\ot\alpha}\Big\|\bigg]\leq C e^{-d(\alpha,\omega)/\zeta},
\end{equation}
or equivalently
\begin{equation}
\mathbb{E}\bigg[\sup_{t\in\mathbb{R}}\Big\|\big(\langle\omega|\otimes \matr{I}_S\big) \matr{P}_{I}(\matr{H}) e^{-i \matr{H} t}\big(|\alpha\rangle \otimes\matr{I}_S\big)\Big\|\bigg]\leq C e^{-d(\alpha,\omega)/\zeta}.\label{DynLoc}
\end{equation}
\end{subequations}
In these expressions $\mathbb{E}[.]$ is the expectation with respect to all the random variables, $\sup_{t\in\mathbb{R}}$ is the supremum over time, $C<\infty$, $\zeta>0$ and $\|.\|$ is the 2-norm. Note, for the choice $S=\emptyset$, $S'=\mathbb{S}$, Eq.~(\ref{DynLocall}) reduces to an exponential bound on off-diagonal matrix elements of the evolution operator.

By construction, the criterion of Eqs.~(\ref{DynLocall}) is a natural extension of the criterion for dynamical localisation of isolated particles. Furthermore, we show below that the criterion is met for a finite amount of disorder in a finite interacting many-body system described by the XYZ-Hamiltonian of Eq.~(\ref{HamiltonianXYZ}) and that once verified, Eqs.~(\ref{DynLocall}) lead to constraints on the magnetisation of $S'$, the correlations between its sites and the linear magnetic response function. 

Before we progress, two remarks are in order. First, for technical reasons, we only consider finite energy intervals $I$. We note however that for any realistic finite system this is not constraining since there must be a finite energy interval containing all possible spectra of the random XYZ model\footnote{We do not assume nor use this in the mathematical proofs.}. Second, from now on we restrict ourselves to sublattices $S'$ comprising only non-adjacent sites. This is not a fundamental requirement of our approach, but it facilitates the mathematical proofs. The case of adjacent sites is discussed in Remark \ref{Snonadjacent}, p. \pageref{Snonadjacent} below.

\section{Dynamical localisation}\label{ManyBodyBound}
In this section we present the main results of this article. For the sake of clarity, we only sketch here the outline of the proof and defer all the details to the appendices. In \S\ref{MainResMB} we give bounds for the onset of dynamical localisation in a many-body system and in \S\ref{locmagnmain}-\ref{loccorrmain} we describe observable signatures of this localisation.
%
%

\subsection{Many-body localisation}\label{MainResMB}~\\

\noindent A major difficulty in the study of many-body localisation as compared to one-body localisation lies in that the effective amount of disorder present in the system grows only \textit{logarithmically} with the size of the relevant space $\G_{S,S'}$. Indeed, since there are only $|S'|$ random $B$-fields affecting the sites of $S'$, there are only $|S'|$ linearly independent configuration potentials. In the same time $S'$ can adopt up to $2^{|S'|}$ configurations.
Thus, the total amount of disorder affecting the SEP is very small as compared to the size of the space on which it evolves. Furthermore, this disorder is also  highly correlated, as the configuration potentials are linearly dependent. 
These observations explain why the proof of localisation for many-body systems is much more involved than its one-body counterpart and leads to much larger bounds on the minimum amount of disorder required for localisation to happen.\\

\noindent \textbf{Outline of our approach:}\\
We obtain a bound on the disorder strength required for the many-body localisation criterion Eq.~(\ref{DynLoc1}) to be met as follows:
\begin{itemize}
\item[1)]We relate Eq.~(\ref{DynLoc1}) to a novel form of fractional moment criterion. More precisely, we consider generalised Green's functions which evolve the SEP on its configuration space. These are matrices dictating the evolution of $S'$ regardless of that of $S$. They arise from a tensor product partition of the system resolvent $(z\matr{I}-\matr{H})^{-1}$. We show that if fractional powers of the 2-norm of these generalised Green's functions are exponentially bounded, then Eq.~(\ref{DynLoc1}) is satisfied.
\item[2)]We express generalised Green's functions non-perturbatively using the method of path-sums, which represents the generalised Green's functions as a superposition of simple paths undergone by the SEP on its configuration space. Simple paths are walks forbidden from visiting any vertex more than once. On Euclidean lattices they are known as self-avoiding walks.
\item[3)]Using path-sums representations, we directly obtain an upper-bound on the required fractional norms. No induction on the number of particles is required, in contrast to Refs. \cite{Aizenman2009} and \cite{Chulaevsky2009}.
\end{itemize}

\noindent These steps finally yield the main theorem of the present article (see Appendix \ref{MBproof} for details):

\begin{theorem}[Many-body dynamical localisation in an interacting system]\label{MBtheorem} Consider an ensemble of particles evolving on a finite lattice according to the XYZ-Hamiltonian of Eq.~(\ref{HamiltonianXYZ}). 
Let $S'$ an ensemble of non-adjacent lattice sites and $I$ be an interval of energy with finite Lebesgue measure $|I|$. Then for any $0<s<1$ and any $2^{-|S'|}<s_1<1$,
\begin{subequations}
\begin{align}
&\mathbb{E}\bigg[\sup_{t\in\mathbb{R}}\Big\|\Big(\matr{P}_{I}(\matr{H}) e^{-i \matr{H} t}\Big)_{\omega\ot\alpha}\Big\|\bigg]\leq C(\omega,\alpha)\,e^{-\frac{(s_1-s)}{s_1(2-s)}\,\frac{d(\alpha,\omega)}{\zeta}},\label{BoundtheoremMB}\\
\shortintertext{where}
&C(\omega,\alpha)=\frac{1}{2\pi}\big(2|I|\big)^{\frac{1}{2-s}} c[s_1]^{\frac{s}{s_1}}\frac{2^{-|S'|s}D^s}{(\sigma_B\sqrt{2\pi})^s}\times\\
&\hspace{35mm}\left(\frac{1}{s |S'|}\Phi\big(e^{-1/\zeta},1,d(\alpha,\omega)+1-s^{-1}\big)\right)^{\frac{s_1-s}{s_1}},\nonumber\\
&\zeta^{-1}=s\log\Big(\frac{\sigma_B\sqrt{2\pi}}{D\mc{J}\,|S'|^{1/s}}\Big),
\end{align}
where $\zeta$ is the localisation length and
\begin{align}
&\zeta>0~~~\text{if and only if}~~~\sigma_B>\sigma_B^{\text{min}}=D\mc{J}\,|S'|^{1/s}/\sqrt{2\pi}.
\end{align}
\end{subequations}
In these expressions, $\mc{J}=(J_x^2+J_y^2)^{1/2}$, $D=2^{|\mathbb{S}|}$, $c[s_1]=4^{1-s_1}k^{s_1}/(1-s_1)$ where $k$ is a finite universal constant and $\Phi$ is the Hurwitz-Lerch zeta-function, which decays algebraically with $d(\alpha,\omega)$ when $d(\alpha,\omega)>s^{-1}-1$.
\end{theorem}

%
%
%
%
%
%
%

\begin{remark}[Dependency on the system size]
The bound we obtain on the critical amount of disorder required for the onset of localisation scales exponentially with the system size $|\mathbb{S}|$.
This is unsurprising in view of the fact that the XYZ Hamiltonian does not conserve the number of particles in the system. For this reason, adding a single site to the lattice doubles the number of configurations available to the system. 
\end{remark}

\begin{remark}[Dependency on the system partition]\label{RemarkPartition}
In the situation where the Hamiltonian is partitioned into its matrix elements, i.e. the flips and statics are the entries of the Hamiltonian $\matr{H}_{\omega\ot\alpha}=(\matr{H})_{\omega\alpha}$, we have $S=\emptyset$, $S'=\mathbb{S}$ and thus $|S'|=|\mathbb{S}|$. This degrades both the bound on the critical disorder $\sigma_B^{\text{min}}$ and the localisation length $\zeta$ linearly with the system size. It is therefore beneficial to use partitions with small $|S'|$ values.
\end{remark}

\begin{remark}[Dependency on the interaction strength along $z$]
The bound obtained above for the critical amount of disorder does not depend on the interaction parameter $\Delta$, but rather holds for any $\Delta$. This is because of a crucial step of our proof, 
where we determine a bound for the norm of a generalised Green's functions  independently of the details of the self-energy. By the same token, we lose information about any influence the interaction strength $\Delta$ might have on the onset of localisation. Thus, our approach does not permit us to study the influence of $\Delta$. We could partially circumvent this problem by treating the deterministic interaction term in configuration $\alpha$ as a realisation of a random variable $\tilde{\Delta}_{\alpha}$ whose probability distribution can be evaluated from the Hamiltonian. Then one would take expectations with respect to a new random configuration potential $\tilde{Y}_{\alpha}= Y_\alpha+\tilde{\Delta}_{\alpha}$ and the strength $\Delta$ of the interaction would play a role similar to the variance of the random $B$-fields. 
\end{remark}

\begin{remark}[Dependency on the lattice]\label{Snonadjacent}
The theorem holds regardless of the dimension and details of the physical lattice. This is because we require $S'$ to comprise only non-adjacent sites. In this situation, the configuration graph $\G_{S,S'}$ must be an $|S'|$-hypercube $\mathcal{H}_{|S'|}$ or a subgraph of it, regardless of the physical lattice. Indeed, $\mathcal{H}_{|S'|}$ connects all configurations of $S'$ differing by one spin flip. Thus if the configuration graph is not $\mathcal{H}_{|S'|}$ or a subgraph of it, then there is at least one transition allowed between two configurations of $S'$ differing by two spin flips. This corresponds to two non-adjacent spins flipping at the same time, which is not allowed by the Hamiltonian. This approach leads to simpler results when deriving observable signatures of the localised regime, see \S\ref{locmagnmain}, \S\ref{loccorrmain} and Appendices \ref{LocMagn} and \ref{Localcorrel}. 
Our results can nonetheless be extended so as to hold in situations where $S'$ comprises adjacent sites. For example, if the physical lattice is $\mathbb{Z}^d$, $d\geq 1$, then the bound on the critical disorder becomes $\sigma_B^{\text{min}}=D\mc{J}\,d^{1/s}|S'|^{1/s}/\sqrt{2\pi}$.
\end{remark}

\begin{remark}[Extension to other Hamiltonians]
While we present Theorem \ref{MBtheorem} explicitly in the case of the XYZ-Hamiltonian, we do so for the sake of concreteness. The arguments developed in the proof of the theorem do not make use of the specifics of the XYZ-Hamiltonian, so that the theorem can easily be extended other Hamiltonians and other probability distributions $\rho$ for the random magnetic fields\footnote{So long as $\rho$ remains H\"{o}lder continuous.}, as considered in Ref. \cite{Aizenman2009}. 
\end{remark}

\begin{remark}[Use of the 2-norm] Theorem \ref{MBtheorem} is stated in terms of the 2-norm which might seem to obscure the physical meaning of the bound Eq.~(\ref{BoundtheoremMB}). Indeed physical quantities of interest such as the system magnetisation and correlations functions are most naturally expressed in terms of the Frobenius norm. Nonetheless we show explicitly in the next section that Theorem \ref{MBtheorem} implies constraints on such physical quantities, as would be expected if the theorem involved the Frobenius norm. 
\end{remark}

%

\subsection{Localisation of sublattice magnetisations}\label{locmagnmain}~\\

\noindent In this section and the next we discuss observable consequences of the criterion for many-body localisation Eq.~(\ref{DynLoc1}). Although the bound for the onset of many-body localisation presented in Theorem \ref{MBtheorem} holds for any finite ensemble of non-adjacent sites $S'$, there is a particular choice of sites for $S'$ that leads to a simple structured configuration graph, allowing for easier interpretations of the consequences of SEP dynamical localisation.

Suppose that $S'$ consists of a set of sites that are well separated in real space. In this case, the surrounding sites belonging to $S$ effectively act as a spin bath with which the sites of $S'$ can freely exchange spin excitations. Clearly, $S$ will act more effectively as a bath in higher dimensions, when each site of the lattice has more nearest neighbours. A key consequence of this choice is that the number of up-spins in $S'$ can vary, even in the XXZ regime $J_x=J_y$ where the total number of up spins in the entire system is conserved. 
Thus, in this situation and regardless of the values of $J_x$ and $J_y$, every one of the $2^{|S'|}$ configurations available to $S'$ is accessible. Furthermore, each configuration of $S'$ is connected to exactly $|S'|$ other configurations, because each spin within $S'$ is free to flip its state when $J_x\neq J_y$, or, in the XXZ regime, by virtue of its contact with the spins of $S$. These considerations uniquely specify the configuration graph of $S'$ as the $|S'|$-hypercube.
Then the dynamics of the SEP is that of a particle undergoing a continuous-time quantum random walk on the hypercube, perceiving highly correlated random potentials at the various sites. 

Now an important consequence of the structure of the hypercube is that the distance $d(\alpha,\,\omega)$ between any two configurations $\alpha$ and $\omega$ available to $S'$ is at least equal to the difference between the number of up-spins in the two configurations $d(\alpha,\,\omega)\geq|n_\alpha-n_\omega|$,
with $n_\alpha$ the number of up-spins in $S'$ when it is in configuration $\alpha$.
Thus SEP dynamical localisation as per our criterion Eq.~(\ref{DynLoc1}) immediately implies localisation in the number of up-spins in the sublattice 
\begin{subequations}
\begin{equation}
\mathbb{E}\bigg[\sup_{t\in\mathbb{R}}\Big\|\big(\matr{P}_{I}(\matr{H})e^{-i\matr{H}t}\big)_{n'\ot n}\Big\|\bigg]\leq C e^{-|n-n'|/\zeta},
\end{equation}
or equivalently
\begin{equation}
\mathbb{E}\Big[\sup_{t\in\mathbb{R}}\|\big(\langle n'|\otimes \matr{I}_S\big)\matr{P}_{I}(\matr{H})e^{-i\matr{H}t}\big(|n\rangle\otimes \matr{I}_S\big)\|\Big]\leq C e^{-|n-n'|/\zeta},
\end{equation}
\end{subequations}
where $|n\rangle$ and $|n'\rangle$ represent any superposition of configurations of $S'$ with exactly $n$ and $n'$ up-spins, respectively.

This yields an observable signature of the many-body localisation of the system: the magnetisation of any sublattice is constrained within an interval centered on its value at time $t=0$. More precisely, let $N_{S'}^\up(t)$ and $N_{S'}^\down(t)$ be the number of up- and down-spins in the sublattice $S'$ at time $t$, respectively. We define the disorder-averaged normalised magnetisation as $M(t)=\mathbb{E}\big[N_{S'}^\up(t)-N_{S'}^\down(t)\big]/|S'|$, with $\mathbb{E}[.]$ the expectation with respect to all random magnetic fields. In Appendix \ref{LocMagn} we show that this quantity is bounded at any time as follows
\begin{figure}[t!]
\vspace{2mm}
\hspace{-7mm}\includegraphics[width=1.105 \textwidth]{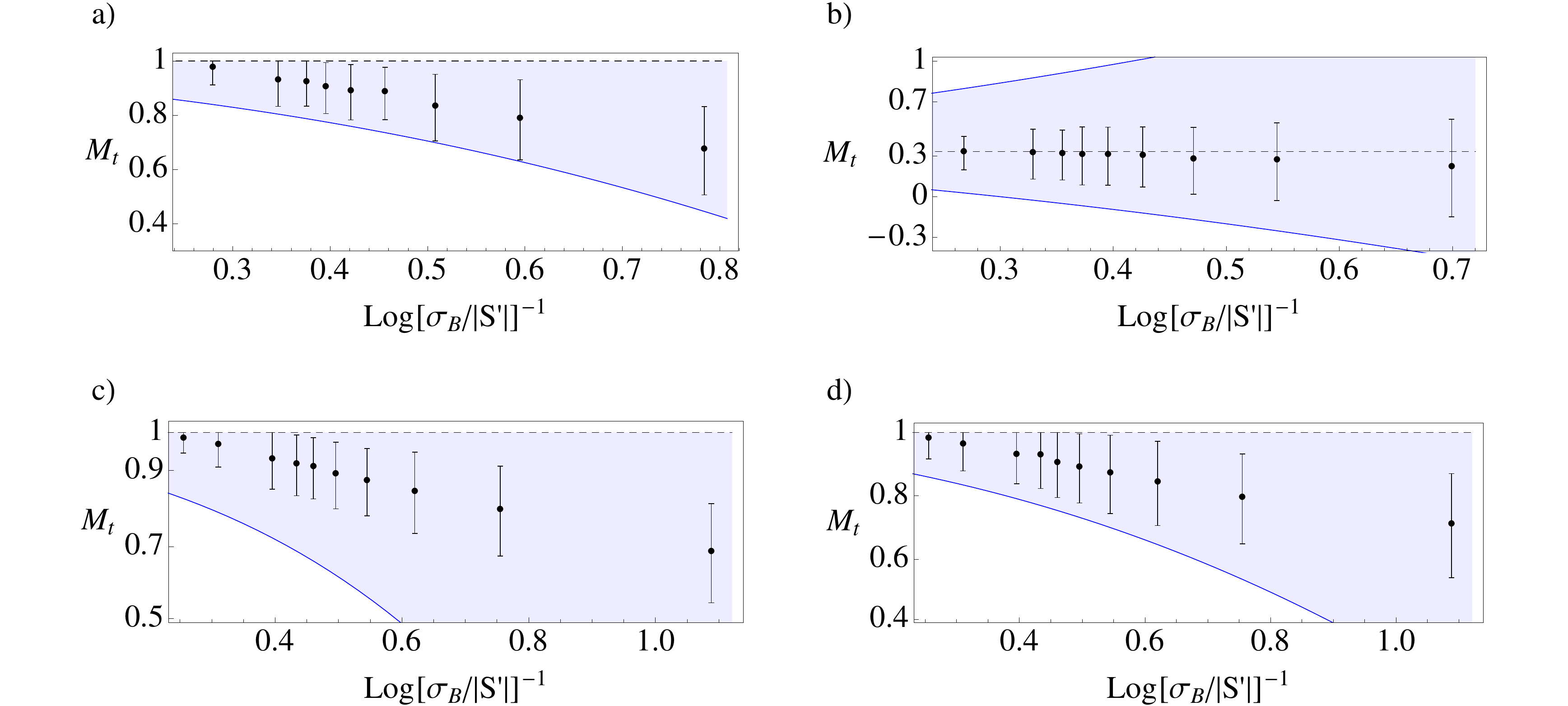}
\vspace{-2mm}
\caption{Data points: disorder-averaged normalised magnetisation $M(t)$ of various sublattices of a 1D lattice with $|\mathbb{S}|=60$ sites. Each error bar represents one standard deviation around the disorder average. Initially, spins located at even and odd sites are up and down, respectively. The sublattices $S'$ shown here comprise: 
\textbf{a)} sites 5 - 13 - 21 - 29 - 37 - 45 - 53, initial magnetisation $M(0)=1$; 
\textbf{b)} sites 13 - 21 - 28 - 40 - 45 - 57, initial magnetisation $M(0)=1/3$; 
\textbf{c)} sites 5  -  11  -  17  -  23  -  29  -  35  -  41  -  47  -  53  -  59, initial magnetisation $M_0=1$;
\textbf{d)} sites 9  -  19  -  29  -  39  -  49  -  59, initial magnetisation $M(0)=1$.
Each data point represents 300 TEBD simulations of $\mathbb{S}$ which is initially populated by 30 up-spins and 30 down-spins. Simulations parameters: XXZ regime $J_x=J_y=1$, $\Delta=1/2$, evolution time: $t=5 J_x^{-1}$, longer evolution times did not yield significantly different results.
In shaded is the region of allowed magnetisation $M(t)$ as per Eq.~(\ref{MagnetoBound}) as a function of $\ln(\sigma_B/|S'|)^{-1}\propto\zeta$. Constants $C$ and $\zeta$ where used as free parameters to fit all results.}\label{MagnetoTEBD}
\vspace{-2mm}
\end{figure}
\begin{figure}[t!]
\vspace{2mm}
\hspace{-7mm}\includegraphics[width=1.105 \textwidth]{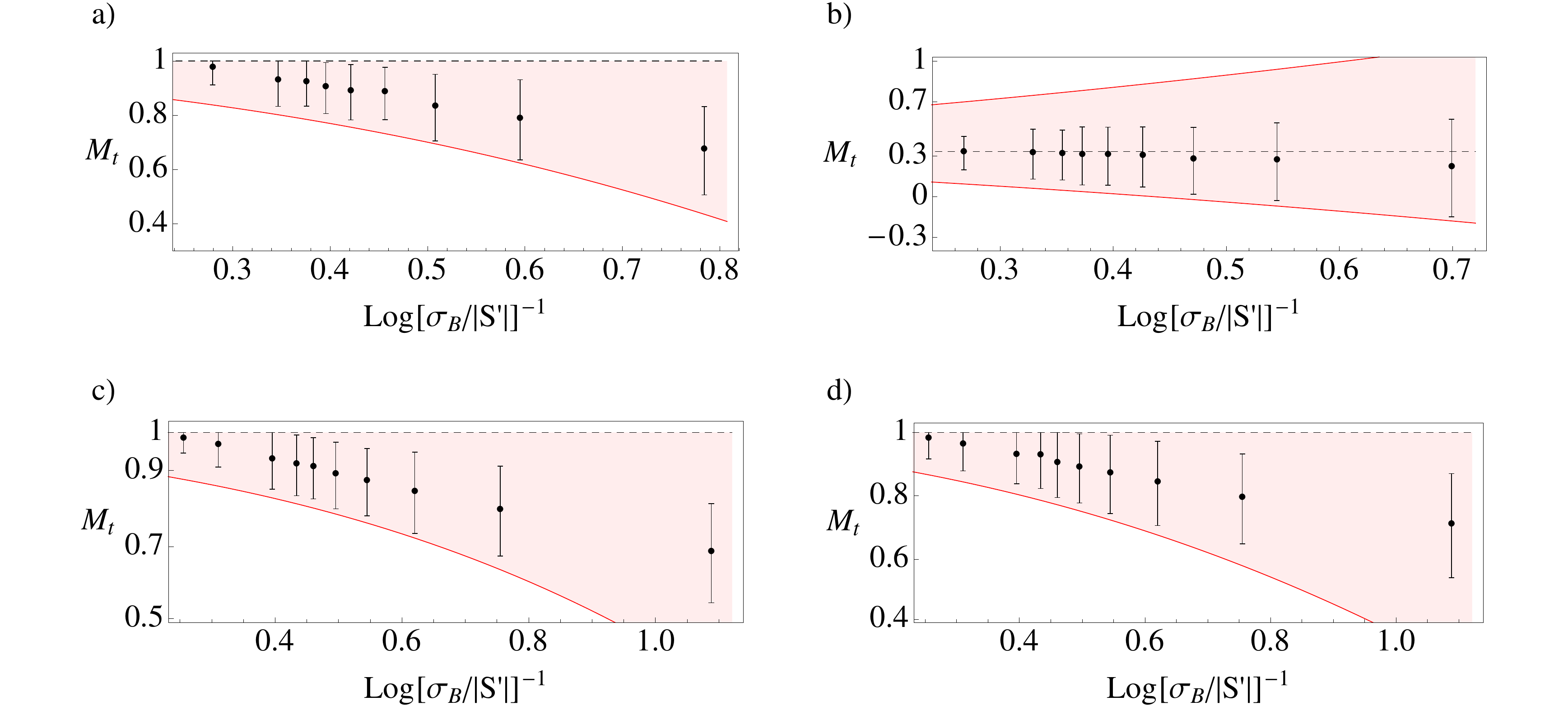}
\vspace{-2mm}
\caption{Same simulations as in Figure \ref{MagnetoTEBD}, but here fitting the values of $C$ and $\zeta$ independently for each figure. The TEBD results are consistent with localisation of the system in its configuration space.}\label{MagnetoTEBD2}
\end{figure}
\begin{equation}\label{MagnetoBound}
1-n^\down(\zeta)\leq M(t)\leq n^\up(\zeta)-1,
\end{equation}
with $n^\up(\zeta)$ and $n^\down(\zeta)$ two time-independent quantities depending on the localisation length $\zeta$ and constant $C$ of Eq.~(\ref{DynLoc1}):
\begin{subequations}
\begin{align}
n^\up(\zeta)&=C^2\sum_{m=0}^{|S'|}\frac{2m}{|S'|}\sum_{d=0}^{|S'|}\mathcal{N}(|S'|,n_0,m,d) e^{-2d/\zeta},\\
n^\down(\zeta)&=C^2\sum_{m=0}^{|S'|}\frac{2m}{|S'|}\sum_{d=0}^{|S'|}\mathcal{N}(|S'|,|S'|-n_0,m,d) e^{-2d/\zeta},
\end{align} 
\end{subequations}
with $n_0=N_{S'}^\up(0)$ the initial number of up-spins in $S'$ and
\begin{align}\label{DefofN2}
&\mathcal{N}(|S'|,n_0,m,d)=\begin{cases}\binom{n_0}{d/2+(n_0-m)/2}\binom{|S'|-n_0}{d/2-(n_0-m)/2},& d+n_0-m~\textrm{even},\\
0,&\textrm{otherwise.}
\end{cases}
\end{align}
In the strongly localised regime, the bound Eq.~(\ref{MagnetoBound}) becomes
\begin{equation}
\lim_{\zeta\to 0}1-n^\down(\zeta)=\lim_{\zeta\to 0}n^\up(\zeta)-1=M(0),
\end{equation}
that is the sublattice magnetisation is constrained around its initial value, as expected. On Figures (\ref{MagnetoTEBD}) and (\ref{MagnetoTEBD2}) we compare the bounds of Eq.~(\ref{MagnetoBound}) with TEBD simulations of a half-filled 1D lattice of 60 sites in the XXZ regime. Our implementation of the TEBD algorithm is based on the open-source Tensor Network Theory (TNT) library \cite{tnt}.
We use $C$ and $\zeta$ as free parameters to fit the simulation results with our theoretical predictions. This means that we do not use their values as predicted by Theorem \ref{MBtheorem}. This is because we could only simulate disorder strengths $\sigma_B$ which are below the bound $\sigma_B^{min}$ obtained in the theorem. 
In Fig. (\ref{MagnetoTEBD}) we determine the values of $C$ and $\zeta$ so as to fit the magnetisations of all observed sublattices at once. Our result are in good agreement with the simulated magnetisations, which demonstrates that the behaviour of this physical quantity is well captured by our bound Eq.~(\ref{MagnetoBound}). In particular, the simulations show that the sublattice magnetisation is strongly constrained around its initial value deep in the localised regime, as predicted by our theoretical bounds. 
In Fig. \ref{MagnetoTEBD2}, the values of $C$ and $\zeta$ are fitted for each sublattice separately. 
The excellent agreement between theory and simulation here shows that the overall dependency of the simulated magnetisation with the disorder is very well captured by our bounds.
%

These results demonstrate the validity of our criterion, which quantitatively predicts a physically observable signature of many-body localisation.
%
%

\subsection{Localisation of sublattice correlations and magnetic susceptibility}\label{loccorrmain}~\\

\noindent Similarly to the sublattice magnetisation presented above, we prove 
in Appendix \ref{Localcorrel} that once the criterion for many-body dynamical localisation Eq.~(\ref{DynLoc1}) is fulfilled, then the correlations between any two sites of the sublattice $S'$ will be localised around their initial value. The correlations at time $t$ between two sites $i$ and $j$ of $S'$ that we consider are given by 
\begin{equation}
\tau_{i,j}(t)=\mathbb{E}\big[\mathrm{Tr}[\matr{P}_I(\matr{H})\,\sigma_z^i(t)\,\matr{P}_I(\matr{H})\,\sigma_z^j \,\rho_{\mathbb{S}}]\big],\label{correla}
\end{equation}
where $\mathbb{E}[.]$ is the expectation with respect to all random magnetic fields, $\rho_{\mathbb{S}}$ is the initial density matrix of the entire system $\mathbb{S}$ and $\sigma_z^i(t)=e^{i\matr{H}t}\sigma_z^ie^{-i\matr{H}t}$. The quantity $\tau_{t}(i,j)$ is of interest because it relates to the disorder-averaged linear magnetic response function $\chi_{i,j}(t)$. More precisely we have
\begin{equation}
\chi_{i,j}(t)=-i\,\mathbb{E}\big[\langle[\sigma_z^i(t),\sigma_z^j(0)]\rangle\big]=-i\big(\tau_{i,j}(t)-\tau_{j,i}(-t)\big),
\end{equation}
which determines the response of the magnetisation at site $i$ to a time-dependent magnetic field applied at site $j$ \cite{forster}.

In the situation where $S'$ is initially in a single configuration $\alpha$ and is dynamically localised as per our criterion Eq.~(\ref{DynLoc1}), we obtain the following bounds for the sublattice correlations
\begin{figure}[t!]
\begin{center}
\vspace{2mm}
\includegraphics[width=.75 \textwidth]{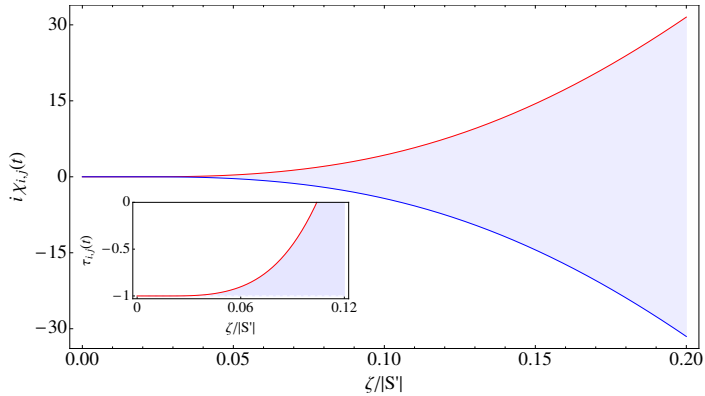}
\end{center}
\vspace{-2mm}
\caption{Bounds on the disorder-averaged magnetic linear response $i\chi_{i,j}(t)$ in the localised regime, as a function of the localisation length $\zeta$. The shaded region is the region allowed by the bound Eq.~(\ref{chibound})
We consider here a sublattice $S'$ comprising $|S'|=10$ sites. Initially, the sublattice is in a configuration $\alpha$ with $|\alpha|=3$ up spins and $i$ and $j$ are up and down in $\alpha$, respectively. In inset: bound on the disorder averaged correlation function $\tau_{i,j}(t)$, initially $\tau_{i,j}(0)=-1$. The shaded region is the region allowed by the bounds Eqs.~(\ref{BoundTau}).}\label{CorrelationsExample}
\end{figure}
\begin{subequations}\label{BoundTau}
\begin{equation}
\begin{cases}
\tau_{i,j}(t)\leq 2\mathcal{K}(|S'|,\,|\alpha|)-1,&\text{$j$ is up in }\alpha,\\
\tau_{i,j}(t)\geq -2\mathcal{K}(|S'|,\,|\alpha|)+1,&\text{$j$ is down in }\alpha,
\end{cases}
\end{equation}
and simultaneously
\begin{equation}
\begin{cases}
\tau_{i,j}(t)\geq -2\mathcal{Q}(|S'|,\,|\alpha|)+1,&\text{$j$ is up in }\alpha,\\
\tau_{i,j}(t)\leq 2\mathcal{Q}(|S'|,\,|\alpha|)-1,&\text{$j$ is down in }\alpha,
\end{cases}
\end{equation}
\end{subequations}
where $|\alpha|$ is the number of up-spins in $S'$ when it is in configuration $\alpha$ and $\mathcal{K}$ and $\mc{Q}$ are given by
\begin{subequations}
\begin{align}
&\hspace{1mm}\mathcal{K}(|S'|,\,|\alpha|)=\nonumber\\
&\hspace{-3mm}\begin{cases}
C^2\sum_{m=1}^{|S'|-1}\sum_{d=0}^{|S'|-1}\hspace{-1mm}\mathcal{N}(|S'|\hspace{-.3mm}-\hspace{-.3mm}1,|\alpha|\hspace{-.3mm}-\hspace{-.3mm}1,m\hspace{-.3mm}-\hspace{-.3mm}1,d)\, e^{-2d/\zeta},&\hspace{-.5mm}\text{$i$ is up in $\alpha$,}\\
C^2\sum_{m=1}^{|S'|-1}\sum_{d=1}^{|S'|-1}\hspace{-1mm}\mathcal{N}(|S'|-1,|\alpha|,m-1,d-1)\, e^{-2d/\zeta},&\hspace{-.5mm}\text{$i$ is down in $\alpha$,}\\
\end{cases}\\
&\hspace{1mm}\mathcal{Q}(|S'|,\,|\alpha|)=\nonumber\\
&\hspace{-3mm}\begin{cases}
C^2\sum_{m=1}^{|S'|-1}\sum_{d=1}^{|S'|-1}\hspace{-1mm}\mathcal{N}(|S'|-1,|\alpha|,m-1,d-1)\, e^{-2d/\zeta},&\hspace{-.5mm}\text{$i$ is up in $\alpha$,}\\
C^2\sum_{m=1}^{|S'|-1}\sum_{d=0}^{|S'|-1}\hspace{-1mm}\mathcal{N}(|S'|\hspace{-.3mm}-\hspace{-.3mm}1,|\alpha|\hspace{-.3mm}-\hspace{-.3mm}1,m\hspace{-.3mm}-\hspace{-.3mm}1,d)\, e^{-2d/\zeta},&\hspace{-.5mm}\text{$i$ is down in $\alpha$.}\\
\end{cases}
\end{align}
\end{subequations}
%
Now let $\tau_{i,j}^-(\zeta)$ and $\tau^+_{i,j}(\zeta)$ designate the bounds of Eqs.~(\ref{BoundTau}) obtained above for $\tau_{i,j}(t)$, that is
\begin{equation}
\tau^-_{i,j}(\zeta)\leq \tau_{i,j}(t)\leq \tau^+_{i,j}(\zeta).
\end{equation}
Then a bound on the disorder-averaged magnetic response function follows straightforwardly,
\begin{equation}\label{chibound}
\tau^+_{j,i}(\zeta)-\tau^-_{i,j}(\zeta)\leq i \chi_{i,j}(t)\leq \tau^+_{i,j}(\zeta)-\tau^-_{j,i}(\zeta).
\end{equation}
In particular, in the strong disorder limit, we find analytically that 
\begin{equation}\label{LimitTau}
\lim_{\zeta\to0}\tau^-_{i,j}(\zeta)=\lim_{\zeta\to0}\tau^+_{i,j}(\zeta)=\lim_{\zeta\to0}\tau_{i,j}(t)=\tau_{i,j}(0),
\end{equation}
and similarly $\lim_{\zeta\to0}\tau_{j,i}(-t)=\tau_{j,i}(0)$. In other terms, the correlations are constrained around their initial value, as expected.
Since furthermore $\tau_{i,j}(0)=\tau_{j,i}(0)$, Eqs.(\ref{chibound}, \ref{LimitTau}) prove that in the strongly localised regime the disorder-averaged magnetic response function vanishes
\begin{equation}
\lim_{\zeta\to 0} i\chi_{i,j}(t)=0.
\end{equation}
This implies that the disorder-averaged magnetic susceptibility $\tilde{\chi}_{i,j}(\omega)$, defined as the temporal Fourier transform of $\chi_{i,j}(t)$, decreases in magnitude as the disorder strength increases, and vanishes in the strongly localised regime. Therefore, in this regime the system displays a negligibly small induced magnetisation in response to a magnetic perturbation applied to a distant region of the lattice.
%
%
Examples of the bounds Eqs.~(\ref{BoundTau}) and Eq.~(\ref{chibound}) as a function of $\zeta$ are shown in Fig.~(\ref{CorrelationsExample}). Here again, our criterion for many-body dynamical localisation implies strong constraints on an observable physical quantity.

\section{Conclusion}\label{conclu}
We have established a natural criterion characterising dynamical localisation in  interacting many-body systems. Focusing on the XYZ Hamiltonian, which effectively does not conserve the number of particles in the system, we have shown explicitly that this criterion is fulfilled for a finite amount of disorder which scales at most exponentially with the system size. 
Furthermore, once fulfilled, the criterion leads to observable signatures of the localisation. In particular, we have bounded the magnetisation of any sublattice of a dynamically localised system and confirmed our predictions with TEBD simulations. Similarly, we have shown that the correlations between spins at the sites of the sublattice as well as the disorder-averaged magnetic linear response function are constrained in the localised regime. 

Our proof relies on a generalisation of the self-avoiding walk representation of the system Green's functions, called path-sums. This representation is independent of the system Hamiltonian or system partition and in fact holds for any matrix and matrix partitions \cite{Giscard2012}. For this reason, the proofs and results presented here in the specific case of the XYZ-Hamiltonian hold in a much more general setting. 
We hope that the path-sums representation, which allows the systematic manipulation of system Green's functions and generalised Green's functions, 
will contribute to the proof of localisation in infinite interacting many-body systems.

The bound for $\sigma_B^{\text{min}}$ we obtain in Theorem \ref{MBtheorem} remains well above the critical disorder observed in simulations, in particular in the XXZ regime where $J_x=J_y$. To tighten the bound further will certainly require the explicit use of characteristic features of the Hamiltonian. Similarly it may not be possible to prove dynamical localisation of infinite interacting many-body systems without using specific properties that only \textit{some} Hamiltonians exhibit. 
In particular, we believe that systems whose spectrum exhibits persisting gaps in the thermodynamic limit are well suited to proofs of dynamical localisation. Indeed in these situations, the moduli of the entries of the evolution operators are expected to \textit{deterministically} exhibit an exponential decay as a function of distance on the configuration-space, i.e. even in the absence of disorder \cite{Benzi1999,Benzi2007,BenziNEW}. While this effect is too weak to cause localisation by itself, it might help control the divergences that appear when bounding the disorder required for the onset dynamical localisation. Interestingly, recent results establish deterministic exponential decay of functions of matrices over $C^\ast$-algebra \cite{BenziManchester}. This implies deterministic bounds for the norms of the generalised Green's functions, which are required to work with non-trivial system partitions $S\neq\emptyset$.

\begin{acknowledgements} 
We thank Gesine Reinert for fruitful discussions.
P.-L. G. is funded by the EPSRC grant EP/K038311/1, Z. C. is funded by a departmental studentship of the Department of Statistics, University of Oxford, M.T. M. acknowledges support from the UK EPSRC via the Controlled Quantum Dynamics CDT. J.J.M.-A. acknowledges Departamento Administrativo de Ciencia, Tecnologa e Innovacion Colciencias for economic support.
\end{acknowledgements}

\appendix

\section{Proof of many-body dynamical localisation (Theorem \ref{MBtheorem})}\label{MBproof}
\subsection{A generalised fractional moment criterion}~\\

\noindent As we have seen in \S\ref{DynSEP} and \S\ref{CritMB}, localisation of the many-body dynamics is equivalent to the dynamical localisation of the SEP on its configuration space. This observation led to the following signature of many-body localisation Eq.~(\ref{DynLoc})
\begin{equation}\label{DynLocAppendix}
\mathbb{E}\Big[\sup_{t\in\mathbb{R}}\big\|\big(\langle\omega|\otimes \matr{I}_S\big) \matr{P}_{I}(\matr{H}) e^{-i \matr{H} t}\big(|\alpha\rangle \otimes\matr{I}_S\big)\big\|\Big]\leq C e^{-d(\alpha,\omega)/\zeta}.
\end{equation}
%
In order to demonstrate that there exists a finite value of the variance $\sigma_B$ of the random magnetic fields such that Eq.~(\ref{DynLocAppendix}) is satisfied, we use a generalisation of the fractional moment criterion originally developed by Aizenman and Molchanov in the one-body setting \cite{Aizenman1993,Aizenman2001}. This criterion  can be summarised as follows. Let $a$ and $b$ be two points of a lattice, then
\begin{equation}\label{ImplicationAizenman}
\begin{cases}
\exists C<\infty,~\zeta>0,~0<s<1~\text{with}\\
\mathbb{E}\Big[\big|G(a,b;z)\big|^s\Big]\leq C e^{-d/\zeta},
\end{cases}
\hspace{-3mm}\Longrightarrow
\begin{cases}
\exists C'<\infty,~\zeta'>0~\text{such that}\\
\mathbb{E}\Big[\sup_{t\in\mathbb{R}}\big|\langle a|\matr{P}_I e^{-i \matr{H} t}|b\rangle\big|\Big]\leq C' e^{-d/\zeta'}.
\end{cases}\nonumber
\end{equation}
where $G(a,b;z)$ is the Green's function evaluated between points $a$ and $b$.
For our purpose it is necessary to extend this criterion to arbitrary system partitions $\mathbb{S}=S\cup S'$. In particular, one must extend the notion of Green's function 
between two points of a lattice to that of Green's functions between two configurations accessible to $S'$. 

This is easily achieved using a tensor-product partition of the system resolvent $\matr{R}_{\matr{H}}(z)=\big(z\matr{I}-\matr{H}\big)^{-1}$. In continuity with existing conventions regarding the Green's functions, we denote the elements of this partition 
\begin{equation}
G_{\G_{S,S'}}(\omega,\alpha;z):=\big(\matr{R}_{\matr{H}}(z)\big)_{\omega\ot\alpha},
\end{equation}
and call them \textit{generalised Green's functions}. These are $2^{|S|}\times 2^{|S|}$ submatrices of the system resolvent $\matr{R}_{\matr{H}}(z)$. 
As required of a tensor-product partition of the resolvent, the generalised Green's functions fulfill the characteristic relation
\begin{subequations}
\begin{equation}
|\omega\rangle\langle\alpha|\otimes G_{\G_{S,S'}}(\omega,\alpha;z)=  \big(\matr{P}_\omega\otimes \matr{I}_S\big)\,.\,\matr{R}_{\matr{H}}(z)\,.\,\big( \matr{P}_{\alpha}\otimes \matr{I}_S\big),\label{CharacteristicTensorPartition}
\end{equation}
with $\matr{P}_\alpha=|\alpha\rangle\langle\alpha|$. Equivalently, 
\begin{equation}
G_{\G_{S,S'}}(\omega,\alpha;z)=  \big(\langle \omega|\otimes \matr{I}_S\big)\,.\,\matr{R}_{\matr{H}}(z)\,.\,\big(| \alpha\rangle\otimes \matr{I}_S\big).
\label{CharacteristicTensorPartition2}
\end{equation}
\end{subequations}
Using generalised Green's functions, Aizenman's fractional moment criterion can be extended to arbitrary system partitions:
\begin{lemma}[Fractional moment criterion for arbitrary system partitions]\label{FtoG}
Let $\mathbb{S}=S\cup S'$ be a system partition, $\G_{S,S'}$ the graph associated to it and $d(\alpha,\,\omega)$ the distance between vertices $v_{\alpha}$ and $v_{\omega}$ on $\G_{S,S'}$. If there exists $C<\infty$ and $\zeta<\infty$ such that, for $0<s<1$, 
\begin{equation}
\mathbb{E}\Big[\big\|G_{\G_{S,S'}}\big(\omega,\alpha;\,z\big)\big\|^s\Big]<Ce^{-d(\alpha,\,\omega)/\zeta},\label{LocalizG}
\end{equation}
then 
\begin{align}
&\mathbb{E}\bigg[\sup_{t\in\mathbb{R}}\Big\|\Big(\matr{P}_{I}(\matr{H}) e^{-i \matr{H} t}\Big)_{\omega\ot\alpha}\Big\|\bigg]\leq \frac{1}{2\pi}\big(2|I|\big)^{\frac{1}{2-s}}\,Ce^{-\frac{d(\alpha,\omega)}{(2-s)\zeta}},
\end{align}
%
where $|I|$ is the Lebesgue measure of $I$.
\end{lemma}
To preserve the flow of the argument, we defer the proof of this result to Appendix  \ref{BoundFracMom}.
Thanks to Lemma \ref{FtoG} above, we may focus our efforts on bounding fractional norms of the generalised Green's functions. We achieve this using their path-sums representation, which we present below.

\subsection{Path-sum representation of the generalised Green's functions}\label{PSGreens}~\\

\noindent To establish localisation of the fractional norms of the generalised Green's functions Eq.~(\ref{LocalizG}), we rely on a generalisation of the self-avoiding walk representation called path-sums. The path-sums representation of any piece of the resolvent $\matr{R}_{\matr{M}}$ of a matrix $\matr{M}$ holds generally for arbitrary partitions of that matrix and continues to hold even if the matrix under consideration has non-commuting elements, e.g. is valued in a division ring. The proof of the most general path-sums representation of matrix functions is given in \cite{Giscard2013} and shall not be reproduced here. The path-sums approach reduces to the self-avoiding walk representation in the one-body situation where the configuration graph of the system is the physical lattice $\mathbb{Z}^d$, but otherwise differs from it. In its most general form, the method of path-sums stems from a fundamental algebraic property of the ensemble of all walks (also known as paths) on (di)graphs: any walk factorises uniquely into products of prime walks, which are the simple paths and simple cycles of $\G$, which on $\mathbb{Z}^2$ are also known as self-avoiding walks and self-avoiding polygons, respectively. The path-sums representation of a matrix resolvent is then obtained from the prime factorisation of the characteristic series of all walks on the (di)graph $\G$. For a full exposition of the algebraic structure of walk ensembles at the origin of the path-sums representation, see \cite{Giscard2012}.

In the case of interest here, the path-sums representation of the generalised Green's function is a superposition of matrix weighted simple paths,
\begin{equation}\label{ODiagG}
G_{\G_{S,S'}}(\omega,\alpha;\,z)=\sum_{\Pi_{\G;\alpha\omega}}\prod_{j=1}^\ell \matr{H}_{\alpha_{j+1}\ot\alpha_j} G_{\G^j_{S,S'}}(\alpha_j,\alpha_j;\,z),
\end{equation}
where $\Pi_{\G;\alpha\omega}$ is the ensemble of simple paths from $v_{\alpha}$ to $v_{\omega}$ on $\G$\footnote{Walks and simple paths are labeled from left to right, here from $\alpha$ to $\omega$, but their weights are matrix products from right to left; hence the notations $\Pi_{\G;\alpha\omega}$ and $G_{\G_{S,S'}}(\omega,\alpha;\,z)$.}, and $(\alpha_1\alpha_2\cdots \alpha_\ell)\in \Pi_{\G;\,\alpha\omega}$ with $\alpha\equiv \alpha_1$ and $\omega\equiv \alpha_\ell$. Recall that a simple path, being a self-avoiding walk, is forbidden from visiting any vertex more than once.
The Green's functions $G_{\G^j_{S,S'}}(\alpha_j,\alpha_j;\,z)$ are entries of the resolvent $\matr{R}_{\matr{H}_{\mathfrak{I}_j}}(z)$, with $\matr{H}_{\mathfrak{I}_j}$ a submatrix of the Hamiltonian $\matr{H}$ obtained upon deleting all the rows and columns of $\matr{H}$ corresponding to configurations $\{\matr{I}_S\otimes |\alpha_i\rangle\}_{i\in \mathfrak{I}_j}$,  $\mathfrak{I}_{j}=\{\alpha_1\cdots \alpha_{j-1}\}$. For the SEP, this is equivalent to living on the subgraph $\G^j_{S,S'}$ of $\G_{S,S'}$ obtained upon deleting vertices $\{v_{\alpha_1},\cdots, v_{\alpha_{j-1}}\}$ of $\G_{S,S'}$. The statics of the resolvent are given by 
\begin{subequations}\label{GSigma}
\begin{equation}
G_{\G_{S,S'}}(\alpha,\alpha;\,z)=\left[z\matr{I}-\matr{H}_{\alpha}-\Sigma_{\G_{S,S'};\,\alpha}\right]^{-1},
\end{equation}
and in general
\begin{equation}
G_{\G^j_{S,S'}}(\alpha_j,\alpha_j;\,z)=\left[z\matr{I}-\matr{H}_{\alpha_j}-\Sigma_{\G^j_{S,S'};\,\alpha_j}\right]^{-1},
\end{equation}
\end{subequations}
with $\Sigma_{\G^j_{S,S'};\,\alpha_j}$ the self-energy of the SEP when located on vertex $v_{\alpha_j}$ of $\G^j_{S,S'}$. It is convenient to make the configuration potentials appear explicitly in Eqs.~(\ref{GSigma}).
Using Eq.~(\ref{StaticFormY}) we have
\begin{equation}\label{eq:formG}
G_{\G^j_{S,S'}}(\alpha_j,\alpha_j;z)=\Big[Y_{\alpha_j}\matr{I}_S+\tilde{\Sigma}_{\G^j_{S,S'};\,\alpha_j}\Big]^{-1},
\end{equation}
where $\tilde{\Sigma}_{\G^j_{S,S'};\,\alpha_j}=z\matr{I}-\tilde{\matr{H}}_{\alpha_j}-\Sigma_{\G^j_{S,S'};\,\alpha_j}$ is an effective self-energy. It depends on the random $B$-fields at sites outside of $S'$. 

We now turn to the properties of the self-energy.
Regardless of the partition, the self-energy has an explicit path-sums representation \cite{Giscard2013}
\begin{equation}
\Sigma_{\G^j_{S,S'};\,\alpha_j}=\sum_{\Gamma_{\G^{j}_{S,S'};\alpha_j}}\prod_{k=1}^{\ell'} \matr{H}_{\mu_{k+1}\ot\mu_k} G_{\G^{j,k}_{S,S'}}(\mu_k,\mu_k;\,z),\label{Sigmaalpha}
\end{equation}
where $\Gamma_{\G^j_{S,S'};\alpha_j}$ is the ensemble of simple cycles, on $\G^j_{S,S'}$ off vertex $\alpha_j$ and $(\mu_1\cdots \mu_{\ell'})\in\Gamma_{\G^j_{S,S'};\alpha_j}$ with $\alpha_j\equiv \mu_1\equiv \mu_{\ell'}$. Note that the Green's functions $G_{\G^{j,k}_{S,S'}}(\mu_k,\mu_k;\,z)$ appearing in the self-energy fulfill Eqs.~(\ref{GSigma}) on $\G^{j,k}_{S,S'}\equiv \G\backslash\{\alpha_1,\cdots\alpha_j,\mu_2,\cdots \mu_k\}$ and thus involve self-energies via Eq.~(\ref{Sigmaalpha}). As an immediate consequence of this observation, all the self-energies are matrix-valued continued fractions.  As long as $|S'|$ is finite, the depth of the continued fraction is finite and equal to $\ell^j_{\alpha_j}+\delta_{\matr{H}_{v_{\ell_p}}\neq 0}$, the length $\ell^j_{\alpha_j}$ of the longest simple path on $\G^j_{S,S'}$ starting at $\alpha_j$ plus one if the last vertex reached by this simple path has a non-zero statics $\matr{H}_{v_{\ell_p}}\neq 0$ \cite{Giscard2012}. 
%
%
%
%
%
%
%
%
%
%


\subsection{Bounding the norms of generalised Green's functions}\label{Statapp}~\\

\noindent Following the result of Lemma \ref{FtoG}, we are left with the task of proving that there exists an amount of disorder $\sigma_B$ such that the fractional norms of generalised Green's functions are exponentially bounded, i.e. 
\begin{equation}\label{BoundGOAL}
\mathbb{E}\left[\big\|G_{\G_{S,S'}}\big(\omega,\alpha;z\big)\big\|^s\right]\leq Ce^{-d(\alpha,\omega)/\zeta},
\end{equation}
with $s\in(0,1)$, $C<\infty$ and most importantly $0<\zeta<\infty$. 

%
%
In the spirit of the proof for the one-body situation \cite{Hundertmark2007}, we bound the expectation of the fractional moments of the generalised Green's functions using their path-sums representation
\begin{align}\label{ODiagGbound}
&\mathbb{E}\Big[\big\|G_{\G_{S,S'}}(\omega,\alpha;\,z)\big\|^s\Big]\nonumber\\
&\hspace{10mm}\leq\mathbb{E}\Big[\sum_{\Pi_{\G;\alpha\omega}}\|G_{\G^\ell_{S,S'}}(\alpha_\ell,\alpha_\ell;\,z)\|^s\prod_{j=0}^{\ell-1} \big\|\matr{H}_{\alpha_{j+1}\ot\alpha_j} G_{\G^j_{S,S'}}(\alpha_j,\alpha_j;\,z)\big\|^s\Big],\nonumber\\
&\hspace{10mm}\leq\sum_{\ell\geq d(\alpha,\omega)}^{L_p} |\Pi_{\G_{S,S'};\alpha\omega}|(\ell)\, (2\mc{J})^{\ell s}\, \mathbb{E}\Big[\prod_{j=0}^\ell \big\|G_{\G^j_{S,S'}}(\alpha_j,\alpha_j;\,z)\big\|^s\Big].
\end{align}
In these expressions we use $\|\matr{H}_{\alpha_{j+1}\ot\alpha_j}\|\leq 2\mc{J}$ for all flips of the Hamiltonian. The quantity $|\Pi_{\G_{S,S'};\alpha\omega}|(\ell)$ is the number of simple paths of length $\ell$ on $\G_{S,S'}$ from $\alpha$ to $\omega$ and $L_p$ is the length of the longest simple path on $\G_{S,S'}$. As long as $|S'|$ is finite, $L_p$ is finite. 
In most cases, the number of simple paths is unknown and we may simply upper bound it by the total number of walks of length $\ell$ on $\G_{S,S'}$ from $\alpha$ to $\omega$, denoted $|W_{\G_{S,S'};\alpha\omega}|(\ell)$, or if this is also unknown by powers of the largest eigenvalue $\lambda_{\G_{S,S'}}$ of $\G_{S,S'}$. 
%
%
In order to progress, it is necessary to identify explicitly the configuration graph $\G_{S,S'}$. 
%
In fact, if $S'$ comprises only non-adjacent sites, then $\G_{S,S'}$ must be the $|S'|$-hypercube $\mc{H}_{|S'|}$. 
\begin{wrapfigure}{r}{0.3\textwidth}
 \vspace{-11.5mm}
  \begin{center}
    \includegraphics[width=0.3\textwidth]{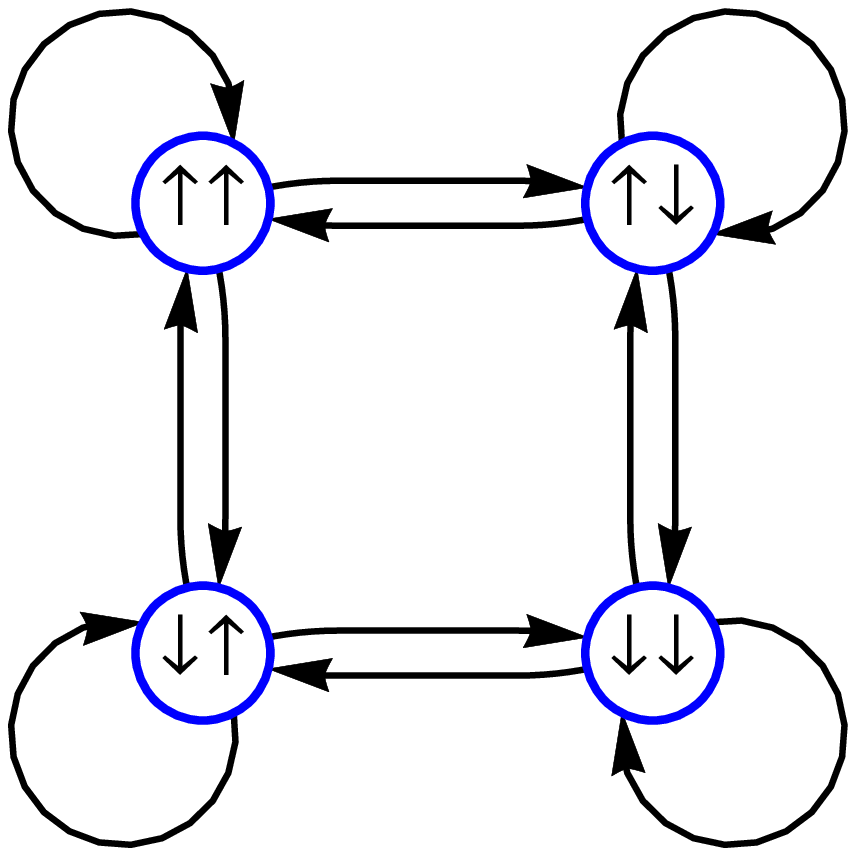}
  \end{center}
  \vspace{-5.5mm}
  \caption{Graph $\mc{H}_2$ \label{figH2}}
  \vspace{-5.5mm}
\end{wrapfigure}
To understand why, remark that the $|S'|$-hypercube connects all configurations of $S'$ differing by one spin flip. For example, the configuration graph $\mathcal{H}_2$ is shown on Fig.~\ref{figH2}. Thus, if $\G_{S,S'}$ is neither the $|S'|$-hypercube nor a subgraph of it, then there is at least one transition allowed between two configurations of the $S'$ differing by two spin flips, e.g. $\up\up\longrightarrow\down\down$. This corresponds to two \emph{non-adjacent} spins flipping at the same time, which is not allowed by the XYZ Hamiltonian. Thus, from now on we shall consider that the configuration space is the $|S'|$-hypercube $\G_{S,S'}\equiv\mc{H}_{|S'|}$, this being the case least favorable to localisation\footnote{The $|S'|$-hypercube being the largest and most connected configuration graph allowed by the Hamiltonian, it has the largest number of simple paths and simple cycles. By Eq.~(\ref{ODiagGbound}), this is the least favorable case for localisation .}. The number of simple paths of length $\ell$ on $\mc{H}_{|S'|}$ is bounded by the number of walks of length $\ell$ on $\mc{H}_{|S'|}$. This is analytically known for any two vertices of any hypercube. For the sake of simplicity however, we use the larger bound $|\Pi_{\G_{S,S'};\alpha\omega}|(\ell)\leq|S'|^\ell$.


As noted in \S\ref{MainResMB}, a major difficulty in the study of many-body localisation as compared to  one-body localisation lies in that even when walking along a simple path on $\G_{S,S'}$, the SEP encounters configuration potentials that are linearly dependent on the ones it previously encountered. 
Indeed, since there are only $|S'|$ random $B$-fields affecting the sites of $S'$, there are only $|S'|$ linearly independent configuration potentials. By contrast, the configuration graph $\G_{S,S'}$ has $2^{|S'|}$ vertices.
A negative consequence of this is that \emph{all} the self-energies entering \emph{all} the generalised Green's functions depend on \emph{all} configurations potentials. Thus, contrary to the one-body case, the generalised Green's functions have a complicated dependency on the configuration potentials. This dependency is analytically accessible for sufficiently small $S'$, e.g. $|S'|=1-6$, yet analytical calculations of the required expectations remain very difficult as we shall see below. 

Another consequence of the linear dependency between the configuration potentials is that we cannot distribute the conditional expectations in the product of generalised Green's functions as is possible in the one-body case. This further complicates the calculation of the required expectation along a simple path $\mathbb{E}\Big[\prod_{j=0}^\ell \big\|G_{\G^j_{S,S'}}(\alpha_j,\alpha_j;\,z)\big\|^s\Big]$. 
In order to recover simpler expressions, we note that the expectation above involves the product measure $d\varrho_{\{B_j\}_{j\notin S'}}\prod_{j=1}^{|S'|}d\varrho_{Y_{\alpha_j}}$, with $\varrho_{Y_{\alpha_j}}$ the cumulative distribution function of $Y_{\alpha_j}$ . Therefore, we can upper bound the expectation of a product of norms of generalised Green's functions by a product of expectations on using the generalised H\"{o}lder inequality \cite{Finner1992}. 
%
%
This leads to 
\begin{align}
&\mathbb{E}\Big[\prod_{j=0}^\ell \big\|G_{\G^j_{S,S'}}(\alpha_j,\alpha_j;\,z)\big\|^s\Big]\leq \prod_{j=0}^\ell\left(\mathbb{E}\Big[ \big\|G_{\G^j_{S,S'}}(\alpha_j,\alpha_j;\,z)\big\|^{q_js}\Big]\right)^{1/q_j},\label{Holderstep}
\end{align}
with $\sum_{j=0}^{\ell} q_j^{-1}=1$. We will see below that the above can only be bounded if $q_j s<1$ for all $j$. Using $q_j=\ell+1$ for all $j$ we thus need $s<1/(\ell+1)$ to guarantee the existence of the above expectations. Since the graph $\G_{S,S'}$ is a subgraph of $\mathcal{H}_{|S'|}$, the longest simple path has length at most $\ell=L_p\leq2^{|S'|}$ and $s<1/(2^{|S'|}+1)$ is sufficient to guarantee existence of all expectations. We will see later that this requirement on $s$ can be waived using a technical result from Ref. \cite{Aizenman2009}.

We now turn to bounding expectations of individual generalised Green's functions. Using conditional expectations we have
\begin{align}
&\mathbb{E}\Big[ \big\|G_{\G^j_{S,S'}}(\alpha_j,\alpha_j;\,z)\big\|^{s_j}\Big]=\mathbb{E}_{\{B_j\}_{j\notin S'},\,\{Y_{\alpha_i}\}_{i\neq j}}\bigg[\mathbb{E}_{Y_{\alpha_j}|}\Big[ \big\|G_{\G^j_{S,S'}}(\alpha_j,\alpha_j;\,z)\big\|^{s_j}\Big]\bigg],\label{ExpecYj}
\end{align}
with $s_j=q_j s$ and $\mathbb{E}_{Y_{\alpha_j}|}[.]$ is the expectation with respect to $Y_{\alpha_j}$ knowing the value of $|S'|-1$ linearly independent configuration potentials affecting $S'$, here denoted $\{Y_{\alpha_i}\}_{i\neq j}$, as well as all magnetic fields affecting sites outside of the sublattice $S'$.

As noted earlier, thanks to path-sum, for sufficiently small $|S'|$ (typically less than 6) the exact dependency of a generalised Green's functions on a certain configuration potential is explicitly accessible. Yet, this dependency is still too complicated to be used to bound the expectation of the Green's function with respect to this configuration potential. To see this, let us consider again the simple example where $S'$ comprises only two sites, denoted 1 and 2. In this situation the configuration graph of $|S'|$ is the square $\mc{H}_2$ and the configurations available to $S'$ are $\up\up$, $\up\down$, $\down\up$ and $\down\down$, see Fig.~(\ref{figH2}). The corresponding configuration potentials are $Y_{\up\up}=B_1+B_2$, $Y_{\up\down}=B_1-B_2$, $Y_{\down\up}=-Y_{\up\down}=-B_1+B_2$ and $Y_{\down\down}=-Y_{\up\up}=-B_1-B_2$. Then
\begin{align}\label{S2G}
G_{\mc{H}_2}(\up\up,\,\up\up;\,z)=\Big[z\matr{I}&\left.-Y_{\up\up}\matr{I}-\tilde{\matr{H}}_{\up\up}
-4J^2\,\Sigma_{\up\down}(Y_{\up\up})-4J^2\,\Sigma_{\down\up}(Y_{\up\up})\right.\nonumber\\
&-16J^4\,\Sigma^{\up\down}_{\down\up}\,\Sigma_{\down\down}^{\up\down}(Y_{\up\up})\,\Sigma_{\up\down}(Y_{\up\up})\nonumber\\
&\left.-16J^4\,\Sigma^{\down\up}_{\up\down}\,\Sigma_{\down\down}^{\down\up}(Y_{\up\up})\,\Sigma_{\down\up}(Y_{\up\up})\right]^{-1},
\end{align}
where $\Sigma_{\up\down}(Y_{\up\up})=\big[z\matr{I}-Y_{\up\down}\matr{I}-\tilde{\matr{H}}_{\up\down}-\Sigma^{\up\down}_{\down\down}(Y_{\up\up})\big]^{-1}$, $\Sigma^{\up\down}_{\down\down}(Y_{\up\up})=\big[z\matr{I}+Y_{\up\up}\matr{I}-\tilde{\matr{H}}_{\down\down}-\Sigma^{\up\down}_{\down\up}\big]^{-1}$ and similarly for $\Sigma_{\down\up}(Y_{\up\up})$ and $\Sigma^{\down\up}_{\down\down}(Y_{\up\up})$. The quantities $\Sigma^{\up\down}_{\down\up}$ and $\Sigma^{\down\up}_{\up\down}$ are analytically available and do not depend on $Y_{\up\up}$. The generalised Green's function $G_{\mc{H}_2}(\up\up,\,\up\up;\,z)$ is therefore an exactly known matrix-valued rational function of $Y_{\up\up}$. For this later reason, it is very difficult to find an analytical upper bound for $\mathbb{E}_{Y_{\up\up}}\big[\|G_{\mc{H}_2}(\up\up,\,\up\up;\,z)\|^s\big]$, in particular a bound which is independent of largely unknown quantities such as the spectrum of the Hamiltonian.

Thus, instead of attempting to bound the expectation of a generalised Green's functions $G_{\G^j_{S,S'}}(\alpha_j,\alpha_j;\,z)$, we remark that it is a submatrix of a larger matrix whose dependency on $Y_{\alpha_j}$ is much simpler. To see this, consider merging all the vertices of $\G^j_{S,S'}$ whose configuration potentials are linearly dependent on $Y_{\alpha_j}$. We call this operation an $Y_{\alpha_j}$-collapse and the so created vertex, the collapsed vertex, which we denote $v_{Y_{\alpha_j}}$. An illustration of this procedure in shown on Fig.~(\ref{CollapseY}). 
\begin{figure*}[t!]
\begin{center}
\vspace{-0mm}
\includegraphics[width=.8 \textwidth]{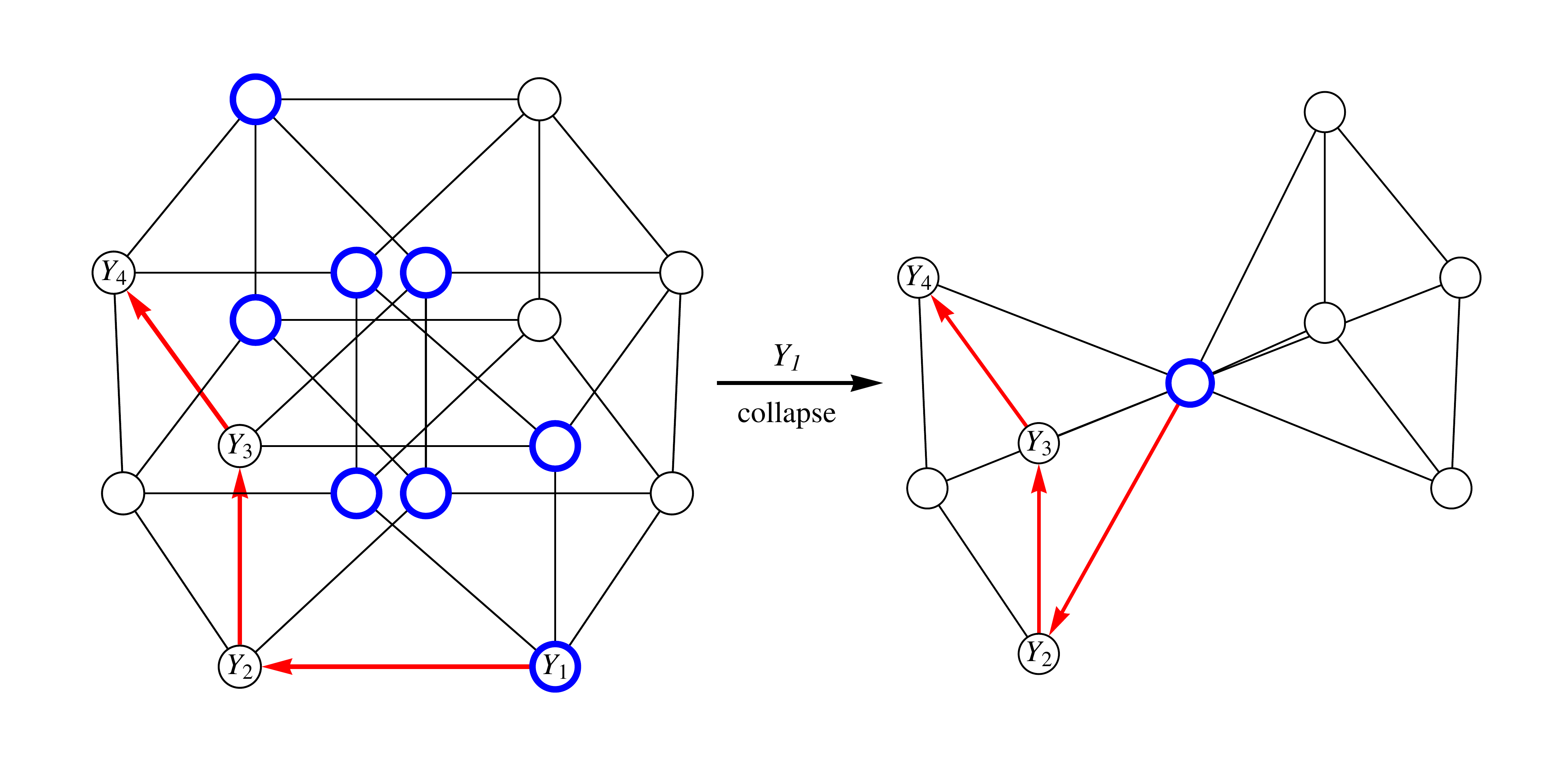}
\end{center}
\vspace{-5mm}
\caption[Collapsing $\mathcal{H}_4$ along configuration potential $Y_1$.]{Left: configuration graph $\mathcal{H}_4$ of a sublattice $S'$ comprising 4 sites. The red arrows represent a simple path along which $S'$ successively encounters the random configuration potentials $Y_{1}\to Y_2\to Y_3 \to Y_4$. In thick blue are the configurations of $S'$ whose random potential is linearly dependent on $Y_1$. Right: configuration graph of the sublattice after collapsing all configurations whose random potential is linearly dependent on $Y_1$ into a single vertex. Note: for clarity the self-loops are not represented on these graphs.}\label{CollapseY}
\vspace{-2mm}
\end{figure*}
The $Y_{\alpha_j}$-collapse gives rise to a new graph, denoted $\G^j_{S,S'}(Y_{\alpha_j})$, from which we immediately obtain 
the generalised Green's function of the collapsed vertex as
\begin{equation}
G_{\G_{S,S'}^j(Y_{\alpha_j})}(v_{Y_{\alpha_j}},v_{Y_{\alpha_j}};\,z)=\left[z\matr{I}-Y_{\alpha_j}\matr{I}-\Sigma_{Y_{v_{\alpha_j}}}\right]^{-1},
\end{equation}
where $\Sigma_{v_{Y_{\alpha_j}}}$ is the self-energy associated to the ensemble of configurations which are linearly dependent on $Y_{\alpha_j}$, i.e. it is
 the sum of all cycles off $v_{Y_{\alpha_j}}$ on $\G^j_{S,S'}(Y_{\alpha_j})\backslash\{v_{Y_{\alpha_j}}\}$. Since all vertices whose configuration potential is linearly dependent on $Y_{\alpha_j}$ have been merged to form $v_{Y_{\alpha_j}}$, none of the remaining vertices of $\G^j_{S,S'}(Y_{\alpha_j})\backslash\{v_{Y_{\alpha_j}}\}$ depend on $Y_{\alpha_j}$. Consequently, $\Sigma_{v_{Y_{\alpha_j}}}$ is independent of $Y_{\alpha_j}$. Furthermore, $G_{\G^j_{S,S'}}(\alpha_j,\alpha_j;\,z)$ is a submatrix of $G_{\G_{S,S'}^j(Y_{\alpha_j})}(v_{Y_{\alpha_j}},v_{Y_{\alpha_j}};\,z)$: more precisely it is the submatrix corresponding to $S'$ being in configuration $\alpha_j$. From these observations we deduce
\begin{equation}
\|G_{\G^j_{S,S'}}(\alpha_j,\alpha_j;\,z)\|^{s_j}\leq\big\|G_{\G_{S,S'}^j(Y_{\alpha_j})}(v_{Y_{\alpha_j}},v_{Y_{\alpha_j}};\,z)\big\|^{s_j},
\end{equation}
and therefore
\begin{equation}
\mathbb{E}_{Y_{\alpha_j}|}\Big[\|G^j_{\G_{S,S'}}(\alpha_j,\alpha_j;\,z)\|^{s_j}\Big]\leq\mathbb{E}_{Y_{{\alpha_j}|}}\Big[\big\|G_{\G_{S,S'}^j(Y_{\alpha_j})}(v_{Y_{\alpha_j}},v_{Y_{\alpha_j}};\,z)\big\|^{s_j}\Big].\label{IndividualBound1}
\end{equation}
The right hand side of the above is bounded thanks to the following lemma, whose proof we defer to Appendix \ref{ProofExpectNorm}:
\begin{lemma}\label{ExpectationProp}
Let $Y$ be a normally distributed random variable with density function $\rho$.
Let $\matr{A}\in\mathbb{C}^{n\times n}$ be a normal matrix. Then the following bound holds for any $0<s<1$,
\begin{equation}
\mathbb{E}_Y\Big[\big\|[Y\matr{I}+\matr{A}]^{-1}\big\|^{s}\Big]\leq\frac{(2n)^s}{1-s}\|\rho\|_{\infty}^s,
\end{equation}
where $\|.\|$ is the 2-norm and $\|.\|_\infty$ the infinity norm.
\end{lemma}
Using the lemma and on assuming that $s_j<1$, we obtain 
\begin{align}
\mathbb{E}_{Y_{{\alpha_j}|}}\Big[\big\|G_{\G_{S,S'}^j(Y_{\alpha_j})}(v_{Y_{\alpha_j}},v_{Y_{\alpha_j}};\,z)\big\|^{s_j}\Big]
&\leq\frac{(2D_j)^{s_j}}{1-s_j}\|\rho_{Y_{{\alpha_j}|}}\|_{\infty}^{s_j}.\label{individualbound}
\end{align}
The quantity $D_j$ is the dimension of $G_{\G_{S,S'}^j(Y_{\alpha_j})}(v_{Y_{\alpha_j}},v_{Y_{\alpha_j}};\,z)$ or, in other terms, $D_j$ is the rank of the perturbation associated with modifications of $Y_{\alpha_j}$. We bound $D_j$ on noting that: i) each vertex of $\G_{S,S'}$ is associated with a vector space of dimension $2^{|S|}$; ii) at most half of the vertices of $\G_{S,S'}$ are linearly dependent on any given configuration potential.
From these observations it follows that $D_j\leq 2^{|S|} \times 2^{|S'|}/2:=D/2$.
Finally, $\rho_{Y_{{\alpha_j}|}}$ is the probability distribution of configuration potential $Y_{\alpha_j}$ knowing the value of $|S'|-1$ linearly independent configuration potentials affecting $S'$ as well as all magnetic fields affecting sites outside of the sublattice $S'$. We show in \S \ref{ConditionalDistr} that this distribution is normal and its variance is at least $\sigma_B$. Taken together, Eqs.~(\ref{IndividualBound1}) and (\ref{individualbound}) thus yield the following bound concerning individual Green's functions
\begin{equation}\label{IndividualBound}
\mathbb{E}_{Y_{\alpha_j}|}\Big[\|G^j_{\G_{S,S'}}(\alpha_j,\alpha_j;\,z)\|^{s_j}\Big]\leq\frac{(2D_j)^{s_j}}{1-s_j}\|\rho_{Y_{{\alpha_j}|}}\|_{\infty}^{s_j}.
\end{equation}

Using the result of Eq.~(\ref{IndividualBound}) in conjunction with Eqs.~(\ref{ODiagGbound}, \ref{Holderstep}, \ref{ExpecYj}), we finally arrive at the desired bound for the disorder-averaged fractional norm of any generalized Green's function
\begin{align}\label{ODiagGbound2}
&\mathbb{E}\Big[\big\|G_{\G_{S,S'}}(\omega,\alpha;\,z)\big\|^s\Big]\leq\hspace{-2mm}\sum_{\ell\geq d(\alpha,\omega)}^{L_p} \hspace{-2.5mm}|\Pi_{\G_{S,S'};\alpha\omega}|(\ell)(2\mc{J})^{\ell s}\prod_{j=0}^\ell\left(\frac{(2D_j)^{s_j }}{1-s_j}\|\rho_{Y_{{\alpha_j}|}}\|_{\infty}^{s_j}\right)^{1/q_j}\hspace{-4mm}.\hspace{-4mm}
\end{align}
We establish in Appendix \ref{ConditionalDistr} that $\|\rho_{Y_{{\alpha_j}|}}\|_{\infty}\leq \|\rho\|_\infty=(\sigma_B\sqrt{2\pi})^{-1}$. Furthermore $|\Pi_{\G_{S,S'};\alpha\omega}|(\ell)\leq |S'|^\ell$, $D_j\leq D/2$ and  
\begin{equation}
\prod_{j=0}^\ell\Big(\frac{1}{1-s_j}\Big)^{1/q_j}=\prod_{j=0}^\ell\Big(\frac{1}{1-(\ell+1)s}\Big)^{1/(\ell+1)}=\frac{1}{1-(\ell+1)s},
\end{equation}
upon choosing $q_j=\ell+1$, which is allowed by the sole requirement on the $q_j$, $\sum_{j=0}^{\ell} q_j^{-1}=1$. Now define $r:=\big(D\mc{J}\|\rho\|_\infty\big)^{s}|S'|$, then the above observations lead to
\begin{align}
\mathbb{E}\Big[\big\|G_{\G_{S,S'}}(\omega,\alpha;\,z)\big\|^s\Big]&\leq  |S'|^{-1}\sum_{\ell\geq d(\alpha,\omega)}^{L_p} \frac{r^{\ell+1}}{1-(\ell+1)s},\nonumber\\
&\leq  |S'|^{-1}\Big|\sum_{\ell\geq d(\alpha,\omega)}^{\infty} \frac{r^{\ell+1}}{1-(\ell+1)s}\Big|.
\label{theseries}
\end{align}
When $r<1$, the series of Eq.~(\ref{theseries}) converges, and we obtain the bound 
\begin{align}
&\mathbb{E}\Big[\big\|G_{\G_{S,S'}}(\omega,\alpha;\,z)\big\|^s\Big]\leq|S'|^{-1}\frac{r^{1/s}}{s}B_{r}\big(d(\alpha,\omega)+1-s^{-1},0\big),\label{EulerBetastep}
\end{align}
with  $B_z(a,b)=\int_0^z t^{a-1}(1-t)^{b-1}dt$ is the incomplete Euler Beta function. To make the bound of Eq.~(\ref{EulerBetastep}) easier to understand, we note that for $0<s<1$, $0<z<1$ and $a\in\mathbb{R}$,
\begin{equation}
z^{1/s} B_{z}(a,0)<z^{(a-1)+1/s} \Phi \left(z,1,a\right),
\end{equation}
with $\Phi$ the Hurwitz-Lerch zeta-function and $\Phi \left(z,1,a\right)$ decays algebraically with $a$ for $a>0$. Thus we find
\begin{subequations}\label{BoundFinalG}
\begin{equation}
\mathbb{E}\Big[\big\|G_{\G_{S,S'}}(\omega,\alpha;\,z)\big\|^s\Big]\leq\frac{1}{s|S'|}\Phi \left(e^{-\frac{1}{\zeta}},1,d(\alpha,\omega)+1-s^{-1}\right)e^{-\frac{d(\alpha,\omega)}{\zeta}}.
\end{equation}
In this expression, we have defined
\begin{align}
&\zeta^{-1}=s\log\Big(\frac{\sigma_B\sqrt{2\pi}}{D\mc{J}\,|S'|^{1/s}}\Big),\\
\shortintertext{i.e. $r=e^{-1/\zeta}$. In particular we have,}
&r>1\iff\zeta>0\iff\sigma_B>D\mc{J}\,|S'|^{1/s}/\sqrt{2\pi}.
\end{align}
\end{subequations} 
At this point, the requirement $s<1/(2^{|S'|}+1)<2^{-|S'|}$, which we used in the proof, clearly seems to heavily degrade the condition on $\sigma_B$ for the onset of localisation: indeed $s<2^{-|S'|}\Rightarrow \sigma_B\geq|S'|^{1/s}\geq |S'|^{2^{|S'|}}$ is extremely large, even for moderate $|S'|$. We circumvent this difficulty by adapting a result of Ref. \cite{Aizenman2009} to our situation:

\begin{lemma}[Based on Aizenman and Warzel \cite{Aizenman2009}]\label{ExtensionLemma} If there exists $s_{max}<1$ such that for all $s$, $0<s<s_{max}$,
\begin{equation}
\mathbb{E}\Big[\big\|G_{\G_{S,S'}}(\omega,\alpha;\,z)\big\|^s\Big]<C e^{-d(\alpha,\omega)/\zeta},
\end{equation}
then for all $s\in(0,1)$ we have
\begin{equation}
\mathbb{E}\Big[\|G_{\G_{S,S'}}(\omega,\alpha;\,z)\|^s\Big]\leq  c[s_1]^{s/s_1}\frac{2^{|S|s}}{(\sigma_B\sqrt{2\pi})^s}
C^{\frac{s_1-s}{s_1}} e^{-\frac{(s_1-s)}{ s_1}\frac{d(\alpha,\omega)}{\zeta}},
\end{equation}
with $s_{max}<s_1<1$ and $c[s_1]=4^{1-s_1}k^{s_1}/(1-s_1)$ with $k$ a finite universal constant.
\end{lemma}
We call this result \textit{the extension lemma}. For the sake of clarity we defer its proof to  Appendix \ref{ProofExtendingLemma}. Using the extension lemma, we may extend the exponential bound from $0<s<1/(2^{|S'|}+1)$ to $0<s<1$. We obtain, for all $0<s<1$,
\begin{align}
&\mathbb{E}\Big[\big\|G_{\G_{S,S'}}(\omega,\alpha;\,z)\big\|^s\Big]\leq c[s_1]^{\frac{s}{s_1}}\frac{2^{-|S'|s}D^s}{(\sigma_B\sqrt{2\pi})^s}\times\label{Efinal}\\
&\hspace{25mm}\left(\frac{1}{s |S'|}\Phi\big(e^{-1/\zeta},1,d(\alpha,\omega)+1-s^{-1}\big)\right)^{\frac{s_1-s}{s_1}}e^{-\frac{(s_1-s)}{s_1}\,\frac{d(\alpha,\omega)}{\zeta}},\nonumber
\end{align}
where we used $2^{|S|}=D\times 2^{-|S'|}$ and the Hurwitz-Lerch zeta-function decays algebraically with $d(\alpha,\omega)$ when $d(\alpha,\omega)>s^{-1}-1$.
The results of Eqs.~(\ref{BoundFinalG}, \ref{Efinal}) together with Lemma \ref{FtoG} establish Theorem \ref{MBtheorem}.\qed

\section{Technical results}\label{TechnicRes}
In this appendix we regroup four technical results: 1) we prove Lemma \ref{FtoG} relating exponential bounds on fractional norms of generalised Green's functions to dynamical localisation; 2) we prove Lemma \ref{ExpectationProp}, which provides an upper bound on expectations of the form $\mathbb{E}_Y\big[\|Y\matr{I}+\matr{A}\|^s\big]$; 3) we prove the extension Lemma \ref{ExtensionLemma}; and 4) we study the marginal and joint distributions of the configuration potentials.

\subsection{Generalised fractional moment criterion}\label{BoundFracMom}~\\

\noindent In this section we establish Lemma \ref{FtoG}, which we restate here for convenience.~\\
%

\noindent \textbf{Lemma \ref{FtoG}.}
\textit{Let $\mathbb{S}=S\cup S'$ be a system partition, $\G_{S,S'}$ the graph associated to it and $d(\alpha,\,\omega)$ the distance between vertices $v_{\alpha}$ and $v_{\omega}$ on $\G_{S,S'}$. If there exists $C<\infty$ and $\zeta<\infty$ such that 
\begin{equation}
\mathbb{E}\Big[\big\|G_{\G_{S,S'}}\big(\omega,\alpha;\,z\big)\big\|^s\Big]<Ce^{-d(\alpha,\,\omega)/\zeta},\label{LocalizG2}
\end{equation}
then for $0<s<1$,
\begin{align}
&\mathbb{E}\bigg[\sup_{t\in\mathbb{R}}\Big\|\Big(\matr{P}_{I}(\matr{H}) e^{-i \matr{H} t}\Big)_{\omega\ot\alpha}\Big\|\bigg]\leq \frac{1}{2\pi}\big(2|I|\big)^{\frac{1}{2-s}}\,Ce^{-\frac{d(\alpha,\omega)}{(2-s)\zeta}},
\end{align}
where $|I|$ is the Lebesgue measure of $I$.}


\begin{proof}
To prove the lemma, we first use Cauchy's integral relation for the matrix exponential \cite{Higham2008} to relate the exponential of the Hamiltonian with its resolvent
\begin{subequations}
\begin{align}
\Big(\matr{P}_{I}(\matr{H}) e^{-i \matr{H} t}\Big)_{\omega\ot\alpha}&=\big(\langle\omega|\otimes \matr{I}_S\big)\matr{P}_{I}(\matr{H}) e^{-i \matr{H} t}\big(|\alpha\rangle\otimes \matr{I}_S\big),\\
&=\big(\langle\omega|\otimes \matr{I}_S\big)\,\frac{1}{2\pi i}\oint_{\gamma_I}\hspace{-1.5mm} e^{-i z t}\, \matr{R}_{\matr{H}}(z) dz\,\big(|\alpha\rangle\otimes \matr{I}_S\big),
\end{align}
\end{subequations}
with $\gamma_I$ a contour enclosing all the eigenvalues $\lambda\in \text{Sp}(\matr{H})\cap I$, with $\text{Sp}(\matr{H})$ the spectrum of $\matr{H}$, and only these. Using a line integral along the contour $\gamma_I$ yields the upper bound
\begin{align}\label{Bound1}
&\sup_{t\in\mathbb{R}}\Big\|\Big(\matr{P}_{I}(\matr{H}) e^{-i \matr{H} t}\Big)_{\omega\ot\alpha}\Big\|\leq \frac{1}{2\pi}\int_{I} \big\|G_{\G_{S,S'}}\big(\omega,\alpha;\gamma_I(x)\big)\big\| \,d\gamma_I,
\end{align}
where $d\gamma_I=\gamma'_I(x) dx$ and we used $\|e^{-i zt}\|_\infty=1$ for $z\in\mathbb{R}$ \footnote{We need only consider $z\in\mathbb{R}$ as all the poles of the resolvent are on the real line so that the contour can be brought arbitrarily close to it.}. More generally, for uniformly bounded functions $f(z)$ we would have $\|f\|_\infty$. In order for the fractional moments to appear, we use H\"{o}lder inequality with $1=s\,a+2(1-a)$, $a=1/(2-s)$ and $s\in(0,1)$ \cite{Hundertmark2000}. This gives
\begin{align}\label{Bound2}
\int_{I} \big\|G_{\G_{S,S'}}\big(\omega,\alpha;\gamma_I(x)\big)\big\| \,d\gamma_I&\leq\left(\int_{I} \big\|G_{\G_{S,S'}}\big(\omega,\alpha;\gamma_I(x)\big)\big\|^2 \,d\gamma_I\right)^{1-a}\times\nonumber\\
&\hspace{8mm}\left(\int_{I} \big\|G_{\G_{S,S'}}\big(\omega,\alpha;\gamma_I(x)\big)\big\|^s \,d\gamma_I\right)^a.
\end{align}
Now we show that the first term of the right hand side is upper bounded by 1. To this end, we note that  $\text{Sp}(\matr{H})$ is discrete and the eigenvalues have almost surely multiplicity one \cite{Hundertmark2000,DelRio1996}. Furthermore, $\matr{H}$ is Hermitian and thus if the domain is finite, the spectral theorem indicates that for all $\lambda\in \text{Sp}(\matr{H})$,  
\begin{equation}
\lim_{\varepsilon\downarrow0} -i\varepsilon\, \matr{R}_{\matr{H}}(\lambda+i\varepsilon)=\matr{P}_\lambda,
\end{equation} 
with $\matr{P}_\lambda=|\varphi_\lambda\rangle\langle \varphi_\lambda|$ the projector onto the eigenstate corresponding to eigenvalue $\lambda$. Therefore, 
\begin{align}
\int_{I} \big\|G_{\G_{S,S'}}\big(\omega,\alpha;\gamma_I(x)\big)\big\|^2 \,d\gamma_I&=\sum_{\lambda\in\text{Sp}(\matr{H})\cap I}\hspace{-1mm}\|\langle \omega| \matr{P}_\lambda\,|\alpha\rangle\|^2\nonumber\\
&=\sum_{\lambda\in\text{Sp}(\matr{H})\cap I}\hspace{-1mm}\|\langle \omega| \matr{P}_\lambda\,|\alpha\rangle\langle \alpha| \matr{P}_\lambda\,|\omega\rangle\|,\nonumber\\
&=\sum_{\lambda\in\text{Sp}(\matr{H})\cap I}\|\langle \omega|\varphi_\lambda\rangle \langle\varphi_\lambda|\alpha\rangle\langle \alpha|\varphi_\lambda\rangle \langle\varphi_\lambda|\omega\rangle\|,
\end{align}
where we used $\|\matr{M}\|^2=\|\matr{M}\matr{M}^\dagger\|$ and $|\alpha\rangle$ and $|\omega\rangle$ are short hand notations for $|\alpha\rangle\otimes \matr{I}_S$ and $|\omega\rangle \otimes \matr{I}_S$, respectively. Now we remark that $\langle\varphi_\lambda|\alpha\rangle\langle \alpha|\varphi_\lambda\rangle$ is a positive real number
which we denote $|\langle\varphi_\lambda|\alpha\rangle|^2$, and similarly $\langle\omega|\varphi_\lambda\rangle\langle\varphi_\lambda|\omega\rangle\equiv|\langle\varphi_\lambda|\omega\rangle|^2\in\mathbb{R}^{+}$. Thus
\begin{align}
\hspace{-3mm}\sum_{\lambda\in\text{Sp}(\matr{H})\cap I}\hspace{-2mm}\|\langle \omega|\varphi_\lambda\rangle \langle\varphi_\lambda|\alpha\rangle\langle \alpha|\varphi_\lambda\rangle \langle\varphi_\lambda|\omega\rangle\|&\leq\sum_{\lambda\in\text{Sp}(\matr{H})\cap I}\hspace{-2mm}|\langle\varphi_\lambda|\omega\rangle|^2\sum_{\lambda'\in\text{Sp}(\matr{H})\cap I}\hspace{-2mm}|\langle\varphi_{\lambda'}|\alpha\rangle|^2,\nonumber\\
&\leq1.
\end{align}
Inserting this result in Eq.~(\ref{Bound2}) and then in Eq.~(\ref{Bound1}) yields the bound
\begin{align}
&\sup_{t\in\mathbb{R}}\Big\|\Big(\matr{P}_{I}(\matr{H}) e^{-i \matr{H} t}\Big)_{\omega\ot\alpha}\Big\|\leq \frac{1}{2\pi}\left(\int_{I} \big\|G_{\G_{S,S'}}\big(\omega,\alpha;\gamma_I(x)\big)\big\|^s\,d\gamma_I\right)^{\frac{1}{2-s}}.
\end{align}
The remaining steps are entirely similar to those exposed in \cite{Hundertmark2000} and  \cite{DelRio1996}. Taking the expectation values on both sides, we note that $1/(2-s)<1$ and H\"{o}lder inequality yields $\mathbb{E}\big[|f|^{1/(2-s)}\big]\leq\big(\mathbb{E}\big[|f|\big]\big)^{1/(2-s)}$. Now Fubini's theorem gives 
\begin{align}
&\mathbb{E}\bigg[\sup_{t\in\mathbb{R}}\Big\|\Big(\matr{P}_{I}(\matr{H}) e^{-i \matr{H} t}\Big)_{\omega\ot\alpha}\Big\|\bigg]\leq\frac{1}{2\pi}\left(\int_{I} \mathbb{E}\Big[\big\|G_{\G_{S,S'}}\big(\omega,\alpha;\gamma_I(x)\big)\big\|^s\Big]\,d\gamma_I\right)^{\frac{1}{2-s}}.
\end{align}
If the fractional moments of the Green's function are bounded by an exponential, e.g. $Ce^{-d(\alpha,\omega)/\zeta}$, the M-L inequality finally implies
\begin{equation}
\mathbb{E}\bigg[\sup_{t\in\mathbb{R}}\Big\|\Big(\matr{P}_{I}(\matr{H}) e^{-i \matr{H} t}\Big)_{\omega\ot\alpha}\Big\|\bigg]\leq C_I\,Ce^{-\frac{d(\alpha,\omega)}{(2-s)\zeta}},
\end{equation}
where $C_I=(2\pi)^{-1}\big(2|I|\big)^{\frac{1}{2-s}}$ since $2|I|$ is the length of the smallest contour enclosing all poles within $I$, where $|I|$ the Lebesgue measure of $I$. This result extends the fractional moment criterion of Aizenman and Molchanov  to arbitrary partitions of the system.\qed
\end{proof}

\subsection{Bounding the expectation of a fractional norm}\label{ProofExpectNorm}~\\

\noindent In this section we establish Lemma \ref{ExpectationProp}, which we restate here for convenience.~\\

\noindent \textbf{Lemma 2.} \textit{Let $Y$ be a normally distributed random variable with distribution function $\rho$.
Let $\matr{A}\in\mathbb{C}^{n\times n}$ be a normal matrix. Then the following bound holds for any $0<s<1$,
\begin{equation}
\mathbb{E}_Y\Big[\big\|[Y\matr{I}+\matr{A}]^{-1}\big\|^{s}\Big]\leq\frac{(2n)^s}{1-s}\|\rho\|_{\infty}^s,
\end{equation}
where $\|.\|$ is the 2-norm.}
\begin{proof}
The lemma follows from the layer-cake integral for the expectation \cite{Aizenman2009}. 
Since $\matr{A}$ is normal and finite, there exist a unitary matrix $\matr{U}$ and diagonal matrix $\matr{D}$ such that $Y\matr{I}+\matr{A}=\matr{U}(Y\matr{I}+\matr{D})\matr{U}^{\dagger}$. 
Now let $\text{Sp}(\matr{A})$ be the spectrum of $\matr{A}$ and suppose $Y$ is fixed. Then 
\begin{align}
\big\|[Y\matr{I}+\matr{A}]^{-1}\big\|&=\max_{\lambda\in \text{Sp}(\matr{A})}|Y+\lambda|^{-1},\nonumber\\
&=\Big(\min_{\lambda\in \text{Sp}(\matr{A})}|Y+\lambda|\Big)^{-1}.\label{Min}
\end{align}
Let $\boldsymbol{1}[x]$ be the indicator function, which equals 1 if assertion $x$ is true and $0$ otherwise. 
Then the expectation is upper bounded as follows
\begin{align}
\mathbb{E}_Y\Big[\big\|[Y\matr{I}+\matr{A}]^{-1}\big\|^{s}\Big]&=\int_{-\infty}^{\infty}\big\|[Y\matr{I}+\matr{A}]^{-1}\big\|^{s}d\varrho,\nonumber\\
&=\int_{-\infty}^{\infty}\int_{0}^{\|[Y\matr{I}+\matr{A}]^{-1}\|^{s}}1dt\, d\varrho,\nonumber\\
&=\int_{-\infty}^{\infty}\int_{0}^{\infty}\boldsymbol{1}\Big[\big\|[Y\matr{I}+\matr{A}]^{-1}\big\|^{s}>t\Big] dt\,d\varrho,
\end{align}
where $\varrho$ is the cumulative distribution function.
By Fubini's theorem and on using Eq.~(\ref{Min}) this is
\begin{align}\label{StepDemoHolder}
\hspace{-3mm}\mathbb{E}_Y\Big[\big\|[Y\matr{I}+\matr{A}]^{-1}\big\|^{s}\Big]&=\int_{0}^{\infty}\int_{-\infty}^{\infty}\boldsymbol{1}\Big[\min_{\lambda\in \text{Sp}(\matr{A})}|Y+\lambda|<t^{-1/s}\Big] d\varrho\,dt,\nonumber\\
&=\int_0^{\infty}\min\Big(1,\int_{-\infty}^{\infty}\boldsymbol{1}\Big[\min_{\lambda\in \text{Sp}(\matr{A})}|Y+\lambda|<t^{-1/s}\Big] d\varrho\Big)\,dt.
\end{align}
Now let $\mathcal{I}_i=[a_i,a_{i+1}]$, $0\leq i\leq n-1$ be $n$ intervals such that $\min_{\lambda\in \text{Sp}(\matr{A})}|Y+\lambda|=|Y+\lambda_i|$ for all $Y\in \mathcal{I}_i$. Then
\begin{align}
\int_{-\infty}^{\infty}\boldsymbol{1}\Big[\min_{\lambda\in \text{Sp}(\matr{A})}|Y+\lambda|<t^{-1/s}\Big] d\varrho&=\sum_{i=0}^n\int_{a_i}^{a_{i+1}}\boldsymbol{1}\big[|Y+\lambda_i|<t^{-1/s}\big] d\varrho,\nonumber\\
&\leq\sum_{i=0}^n\int_{-\infty}^{\infty}\boldsymbol{1}\big[|Y+\lambda_i|<t^{-1/s}\big] d\varrho,\nonumber\\
&=\sum_{i=0}^n\varrho\big(|Y+\lambda_i|<t^{-1/s}\big),\nonumber\\
&\leq 2n \|\rho\|_\infty t^{-1/s},\label{laststep}
\end{align}
where, by virtue of the fact that $\varrho$ is H\"{o}lder continuous\footnote{We consider only on the case where $\rho$ is H\"{o}lder continuous of order 1, which is the case of the normal distribution. The proof extends well to H\"{o}lder continuous distributions of higher order \cite{Hundertmark2000}.}, for any interval $[a,\,b]$, $\varrho([a,\,b])\leq |a-b|\,K$ with $K=\sup_{a,\,b}\varrho([a,\,b])/|a-b|\leq2\|\rho\|_\infty$. The last step of Eq.~(\ref{laststep}) is a rather crude upper bound which may be refined if more is known about the spectrum of $\matr{M}$. 
Inserting this result in Eq.~(\ref{StepDemoHolder}) gives
\begin{align}
\mathbb{E}_Y\Big[\big\|[Y\matr{I}+\matr{A}]^{-1}\big\|^{s}\Big]&\leq \int_0^{\infty}\min\Big(1,2n\|\rho\|_\infty t^{-1/s}\Big)\,dt,\nonumber\\
&= \int_0^{m}1dt+2n\|\rho\|_\infty\int_{m}^{\infty} t^{-1/s}dt,\nonumber\\
&= \frac{n^s}{1-s}2^s\|\rho\|^s_\infty.
\end{align}
We obtain the last line on choosing the $m$ that minimises the bound.\qed
\end{proof}

\subsection{Proof of the extension lemma}\label{ProofExtendingLemma}~\\

\noindent In this section we prove the extension Lemma \ref{ExtensionLemma}, which we restate here for convenience.\\

\noindent \textbf{Lemma \ref{ExtensionLemma}. (Extension lemma, based on Aizenman and Warzel \cite{Aizenman2009})}\textit{ If there exists $s_{max}<1$ such that for all $s$, $0<s<s_{max}$,
\begin{equation}
\mathbb{E}\Big[\big\|G_{\G_{S,S'}}(\omega,\alpha;\,z)\big\|^s\Big]<C e^{-d(\alpha,\omega)/\zeta},
\end{equation}
then for all $\tau\in(0,1)$ we have
\begin{equation}
\mathbb{E}\Big[\|G_{\G_{S,S'}}(\omega,\alpha;\,z)\|^\tau\Big]\leq  c[s_1]^{\tau/s_1}2^{|S|\tau}\|\rho\|_{\infty}^{\tau}C^{\frac{s_1-\tau}{s_1}} e^{-\frac{(s_1-\tau)d(\alpha,\omega)}{ s_1\zeta}},
\end{equation}
with $s_{max}<s_1<1$.
}
\begin{proof}
Let $0<s_2<s_{\text{max}}<\tau<s_1<1$. Then since $\tau=s_1(\tau-s_2)/(s_1-s_2)+s_2(s_1-\tau)/(s_1-s_2)$, it follows that
\begin{align}
\mathbb{E}\Big[\|G_{\G_{S,S'}}(\omega,\alpha;\,z)\|^\tau\Big]&\\
&\hspace{-10mm}=\mathbb{E}\Big[\|G_{\G_{S,S'}}(\omega,\alpha;\,z)\|^{s_1\frac{\tau-s_2}{s_1-s_2}}\|G_{\G_{S,S'}}(\omega,\alpha;\,z)\|^{s_2\frac{s_1-\tau}{s_1-s_2}}\Big].\nonumber
\end{align}
We use H\"{o}lder inequality to separate the above expectation into a product involving two expectations. Let $q_1=(s_1-s_2)/(\tau-s_2)$ and $q_2=(s_1-s_2)/(s_1-\tau)$ and remark that $q_1^{-1}+q_2^{-1}=1$. Thus
\begin{align}
\mathbb{E}\Big[\|G_{\G_{S,S'}}(\omega,\alpha;\,z)\|^\tau\Big]&\leq\mathbb{E}\Big[\|G_{\G_{S,S'}}(\omega,\alpha;\,z)\|^{q_1 s_1\frac{\tau-s_2}{s_1-s_2}}\Big]^{1/q_1}\times\nonumber\\
&\hspace{25mm}\mathbb{E}\Big[\|G_{\G_{S,S'}}(\omega,\alpha;\,z)\|^{q_2 s_2\frac{s_1-\tau}{s_1-s_2}}\Big]^{1/q_2},\nonumber\\
&=\mathbb{E}\Big[\|G_{\G_{S,S'}}(\omega,\alpha;\,z)\|^{s_1}\Big]^{\frac{\tau-s_2}{s_1-s_2}}\times\nonumber\\
&\hspace{25mm}\mathbb{E}\Big[\|G_{\G_{S,S'}}(\omega,\alpha;\,z)\|^{s_2}\Big]^{\frac{s_1-\tau}{s_1-s_2}}.
\end{align}
Now $\mathbb{E}\big[\|G_{\G_{S,S'}}(\omega,\alpha;\,z)\|^{s_2}\big]\leq C e^{-d(\alpha,\omega)/\zeta}$ is exponentially bounded since $s_2<s_{\text{max}}$. Furthermore, $\mathbb{E}\big[\|G_{\G_{S,S'}}(\omega,\alpha;\,z)\|^{s_1}\big]$ is bounded as well. This follows from Theorem 2.1 of Ref. \cite{Aizenman2009}, based on results of Ref. \cite{Aizenman2006}. These results give
\begin{equation}
\|G_{\G_{S,S'}}(\omega,\alpha;\,z)\|^{s_1}\leq c[s_1]\big(2^{|S|}\big)^{s_1}\|\rho\|_{\infty}^{s_1}
\end{equation}
with $c[s_1]=4^{1-s_1}k^{s_1}/(1-s_1)$ and $k$ a finite universal constant.
%
%
%
%
%
Then,
\begin{align}
\mathbb{E}\Big[\|G_{\G_{S,S'}}(\omega,\alpha;\,z)\|^\tau\Big]\leq \Big(c[s_1]2^{|S|s_1}\|\rho\|_{\infty}^{s_1}\Big)^{\frac{\tau-s_2}{s_1-s_2}}C^{\frac{s_1-\tau}{s_1-s_2}} e^{-\frac{(s_1-\tau)d(\alpha,\omega)}{\zeta(s_1-s_2)}}.
\end{align}
Since we may choose $s_2$ to be arbitrarily small, we obtain for all $0<\tau<1$,
\begin{align}
\mathbb{E}\Big[\|G_{\G_{S,S'}}(\omega,\alpha;\,z)\|^\tau\Big]\leq c[s_1]^{\tau/s_1}2^{|S|\tau}\|\rho\|_{\infty}^{\tau}C^{\frac{s_1-\tau}{s_1}} e^{-\frac{(s_1-\tau)d(\alpha,\omega)}{ s_1\zeta}}.
\end{align}
This establishes the lemma. \qed
\end{proof}

\subsection{Conditional distribution of the configuration potentials}\label{ConditionalDistr}~\\

\noindent In this section we study the marginal and joint distribution of the configuration potentials. 

We begin our study by determining the covariance between any two configuration potentials. To this end, we first introduce some notation.
Let $\mathbf{B}=(B_1,\cdots,B_{|S'|})$ be the vector of the random $B$-fields at the sites comprised in $S'$. The $\{B_i\}_{i\in S'}$ are iid normal random variables with variance $\textrm{Var}(B_i)=\sigma_B^2$. Let $\mathbf{v_\alpha}=(v^{(\alpha)}_1,\,v^{(\alpha)}_2,\cdots, v^{(\alpha)}_{|S'|})$ be the vector of coefficients  $a_i^{(\alpha)}=\pm1$ such that
\begin{equation}
Y_\alpha=\mathbf{v_\alpha}\,.\,\mathbf{B}=\sum_{i=1}^{|S'|}v^{(\alpha)}_iB_i,
\end{equation} 
i.e. $v_i^{(\alpha)}=1$ when the spin at site $i$ is up along $z$ in the configuration $\alpha$ of $S'$ and $v_i^{(\alpha)}=-1$ otherwise. Because they are sums of the same random $B$-fields, the configuration potentials are strongly dependent. The following proposition quantifies the covariance between the configuration potentials.

\begin{proposition}\label{CovarianceYY}
Let $Y_\alpha=\mathbf{v_\alpha}\,.\,\mathbf{B}$, $Y_{\beta}=\mathbf{v_{\beta}}\,.\,\mathbf{B}$ and $D_{\alpha,\beta}$ be the Hamming distance between $\mathbf{v_\alpha}$ and $\mathbf{v_{\beta}}$.
Then the covariance $\mathrm{Cov}(Y_\alpha, Y_{\beta})$ is
\begin{equation}\label{CovYY}
\mathrm{Cov}(Y_\alpha, Y_{\beta})=(|S'|-2D_{\alpha,\beta})\sigma_B^2. 
\end{equation}
\end{proposition} 

\begin{proof}
We prove the proposition by expanding the configuration potentials in the definition of the covariance, 
\begin{align}
&\textrm{Cov}(Y_\alpha, Y_{\beta})=\mathbb{E}[Y_\alpha Y_{\beta}]-\mathbb{E}[Y_\alpha]\mathbb{E}[Y_{\beta}],\\
&\hspace{5mm}=\mathbb{E}\bigg[\hspace{-.5mm}\Big(\hspace{-4mm}\sum_{\substack{i\\v_i^{(\alpha)}=v^{(\beta)}_i}}\hspace{-3mm}v_i^{(\alpha)}B_i+\hspace{-3.5mm}\sum_{\substack{i\\v_i^{(\alpha)} \neq v_i^{(\beta)}}}\hspace{-3mm}v^{(\alpha)}_iB_i\Big)\Big(\hspace{-4mm}\sum_{\substack{{i}\\v_i^{(\alpha)}=v_{i}^{(\beta)}}}\hspace{-3mm}v_{i}^{(\beta)}B_{i}+\hspace{-3.5mm}\sum_{\substack{{i}\\v_i^{(\alpha)} \neq v_{i}^{(\beta)}}}\hspace{-3mm}v^{(\beta)}_{i}B_{i}\Big)\hspace{-.5mm}\bigg]\hspace{-.5mm}-\hspace{-.3mm}\mathbb{E}[Y_\alpha]\mathbb{E}[Y_{\beta}].\nonumber
\end{align}
Since the $v_i^{(\alpha,\,\beta)}$ coefficients are $\pm1$, the expectation of the product  is of the form $\mathbb{E}[(a+b)(a-b)]$ and simplifies to
\begin{align}
\textrm{Cov}(Y_\alpha, Y_{\beta})&=\mathbb{E}\bigg[\Big(\hspace{-1.5mm}\sum_{\substack{i\\v^{(\alpha)}_i=v^{(\beta)}_{i}}}\hspace{-1mm}v^{(\alpha)}_iB_i\Big)^2\bigg]-\mathbb{E}\bigg[\Big(\hspace{-1.5mm}\sum_{\substack{i\\v^{(\alpha)}_i \neq v_{i}^{(\beta)}}}\hspace{-1mm}v^{(\alpha)}_iB_i\Big)^2\bigg]-\mathbb{E}[Y_\alpha]\mathbb{E}[Y_{\beta}],\nonumber\\
&=(|S'|-j_\alpha-j_{\beta}+2l_{\alpha,\beta})\sigma_B^2+(|S'|-j_\alpha-j_{\beta})^2\mu^2\nonumber\\
&\hspace{4mm}-(j_\alpha+j_{\beta}-2l_{\alpha,\beta})\sigma_B^2-(j_\alpha-j_{\beta})^2\mu^2-(|S'|-2j_\alpha)\mu(|S'|-2j_{\beta})\mu,\nonumber\\
&=(|S'|-2D_{\alpha,\beta})\sigma_B^2,
\end{align}
where $j_\alpha=\#\{v^{(\alpha)}_i:\,v^{(\alpha)}_i<0\}_{1\leq i\leq |S'|}$ is the number of negative entries in $\mathbf{v_{\alpha}}$ and $l_{\alpha,\beta}=\#\{v^{(\alpha)}_i:\,v^{(\alpha)}_i<0,~v^{(\alpha)}_i=v^{\beta}_i\}_{i\in S'}$ is the number of negative entries $v_{i}^{(\alpha)}$ such that $v_{i}^{(\alpha)}=v_{i}^{(\beta)}$. 
Finally, we observe that $\textrm{Cov}(Y_\alpha, Y_{\beta})\neq0$ and $Y_\alpha$ and $Y_{\beta}$ are dependent, unless $|S'|$ is even and $D_{\alpha,\beta}=|S'|/2$. In this special case and if both $Y_\alpha$ and $Y_\beta$ are normally distributed, then they are independent. \qed
\end{proof}

\begin{lemma}[Conditional distributions of the configuration potentials]\label{StatisticsLemma} Let $\mathbb{S}=S\cup S'$ be a partition of the system and $\mc{Y}=\big\{Y_{\alpha_j}=\sum_{i\in S'}B_i(1-2\delta_{\down_i,\,\alpha_j})\big\}_{1\leq j\leq 2^{|S'|}}$ the ensemble of random configuration potentials affecting $S'$. For a configuration potential $Y_\alpha\in\mc{Y}$, let $E_\alpha\subset \mc{Y}$ be an ensemble of $|S'|-1$ configuration potentials, excluding $Y_{\alpha}$.
Then the configuration potentials of $\{Y_\alpha\}\cup E_\alpha$ are marginally and jointly normal. Furthermore, the variance of the conditional distribution of $Y_{\alpha}$ knowing all configuration potentials of $E_\alpha$ is at least $\sigma_B$.
\end{lemma}

\begin{proof}
In the situation of interest, all magnetic fields affecting $S'$, $\{B_i\}_{i\in S'}$ are iid normal and are thus also jointly normal. Consequently, the configuration potentials, being linear superpositions of magnetic fields, are marginally and jointly normal. Therefore, they form a multi-variate normal distribution (MVN) \cite{Johnson2007}, denoted 
\begin{equation}
\bf{Y}=C\,.\,B \sim \text{MVN}(C\,.\,\mathbb{E}[\bf{B}],C\,.\,\textrm{Cov}(\bf{B})\,.\,C^\dagger),
\end{equation}
with $\bf{B}$ the random vector of the magnetic fields, $\bf{Y}$ the random vector of configuration potentials $Y_\alpha$ and $\bf{C}=\big(v_{\alpha_1},\cdots, v_{\alpha_\ell}\big)^{\mathrm{T}}$ is the coefficient matrix. This establishes that $Y_{\alpha}|\,E_{\alpha}$ is normal.
Thanks to the results of Proposition \ref{CovarianceYY}, the covariance matrix\footnote{Unfortunately, conventions require that we denote both the self-energy and the covariance matrix using a sigma. To differentiate the two, we use a bold sigma for the covariance matrix.} $\boldsymbol{\Sigma}$ of $\{Y_\alpha\}\cup E_\alpha$ is known 
\begin{equation}
\boldsymbol{\Sigma}_{\beta,\gamma}=(|S'|-2D_{\beta,\gamma})\sigma_B^2,\hspace{20mm}Y_{\beta},\,Y_\gamma\,\in \{Y_\alpha\}\cup E_\alpha,
\end{equation}
and the variance $\textrm{Var}(Y_{\alpha}|\,E_{\alpha})=\sigma^2_{Y_{\alpha}|\,E_{\alpha}}$ of $Y_{\alpha}$ knowing $E_{\alpha}$ follows as 
\begin{equation}
\sigma^2_{Y_{\alpha}|\,E_\alpha}=\Big((\boldsymbol{\Sigma}^{-1})_{\alpha,\alpha}\Big)^{-1}.
\end{equation}
A direct calculation shows that this is at least $\sigma_B^2$, which is the variance of a single magnetic field.
This establishes the lemma.\qed
\end{proof}

\section{Localisation of sublattice magnetisation}\label{LocMagn}
Consider the situation where the SEP is dynamically localised on $\G_{S,S'}$. As we have seen, a finite amount of disorder is always sufficient for this to occur on finite systems. Furthermore, in the situation where sites of $S'$ are well separated throughout the physical lattice $\mathcal{L}$, $\G_{S,S'}$ is the $|S'|$-hypercube and otherwise, $\G_{S,S'}\subseteq \mc{H}_{|S'|}$.
In this section we show that SEP dynamical localisation implies localisation of the magnetisation of the sublattice $S'$. This is because, on the hypercube, the distance between two configurations is always greater than or equal to the difference between their number of up-spins. Thus, if the SEP is constrained to stay within a certain distance of its initial configuration, the number of up-spins is also constrained around its initial value. 

\subsection{Localisation of the magnetisation: simple arguments}~\\

\noindent To formalise these observations, let $|n\rangle=\sum_{\alpha_n}c_{\alpha_n}|\alpha_n\rangle$, $c_{\alpha_n}\in\mathbb{C}$, $\sum_{\alpha_n}|c_{\alpha_n}|^2= 1$, be an arbitrary superposition of configurations  $\alpha_n$ of $S'$ with exactly $n$ up-spins. 
Similarly let $|m\rangle$ be an arbitrary superposition of configurations of $S'$ with exactly $m$ up-spins.
Then
\begin{align}
&\mathbb{E}\Big[\sup_{t}\big\|\langle n | \matr{P}_I e^{-i\matr{H}t}|m\rangle\big\|\Big]=\sum_{\alpha_n,\,\alpha_m}c^\ast_{\alpha_n}c_{\alpha_m}\mathbb{E}\big[\sup_{t}\|\langle \alpha_n|\matr{P}_I e^{-i\matr{H}t}|\alpha_m\rangle\|\big].
\end{align}
In the localised regime  $\sup_{t}\|\langle \alpha_n|\matr{P}_I e^{-i\matr{H}t}|\alpha_m\rangle\|\leq Ce^{-d(\alpha_n,\,\alpha_m)/\zeta}$ and since on the hypercube $d(\alpha_n,\,\alpha_m)\geq |n-m|$, we have
\begin{align}
&\mathbb{E}\Big[\sup_{t}\|\langle n | \matr{P}_I e^{-i\matr{H}t}|m\rangle\big\|\Big]\leq C'e^{-|n-m|/\zeta},\\
\shortintertext{where}
&C'=\Big|C\sum_{\alpha_n,\,\alpha_m}c^\ast_{\alpha_n}c_{\alpha_m}\Big|\leq C.
\end{align}
This simple approach indicates that, in the localised regime, the fraction of up-spins present at any time in the sublattice $S'$ is constrained close to its initial value. Below we derive precise upper and lower bounds for the sublattice magnetisation as a function of the localisation length of $S'$.

\subsection{Precise bounds on the magnetisation}~\\

\noindent In this section we derive bounds for the expected magnetisation $M_t$ of $S'$ at time $t$. To do so we find a lower bound for the disorder-averaged fraction of up-spins present in $S'$ at time $t$, which is $\mathbb{E}\big[N_{S'}^\up(t)\big]=\mathbb{E}\big[\mathrm{Tr}[\matr{P}_I\matr{N}^\up(t)\matr{P}_I\,\rho_{\mathbb{S}}]\big]$, with $\rho_{\mathbb{S}}$ the  density matrix of the system and $N_{S'}^\up(t)$ the number of up spins in $S'$ at $t$. For simplicity we take it to be of the form $\rho_{\mathbb{S}}=|\psi_{S'}\otimes \varphi_{S'}\rangle\langle\psi_{S'}\otimes\varphi_{S'}|$. 
The up-spin fraction operator on $S'$ is
\begin{equation}
\matr{N}^\up(t)=e^{i\matr{H}t}\,\bigg(\sum_{n=0}^{|S'|}\lambda_m|m\rangle\langle m|\otimes\matr{I}_S\bigg)\,e^{-i\matr{H}t},
\end{equation}
where $\lambda_m=2m/|S'|$ and $|m\rangle=\sum_{\alpha_m}|\alpha_m\rangle$ with $\alpha_m$ any configuration of $S'$ with exactly $m$ up-spins.
Using cyclicality of the trace, the expected fraction of up-spins at $t$ is therefore of the form
\begin{align}
&\mathbb{E}\big[N_{S'}^\up(t)\big]=\mathbb{E}\bigg[\sum_{m}\lambda_m\mathrm{Tr}\big[\langle\psi_{S'}\otimes \varphi_{S'}|\matr{P}_Ie^{i\matr{H}t}|m\rangle\langle m|e^{-i\matr{H}t}\matr{P}_I|\psi_{S'}\otimes \varphi_{S'}\rangle\big]\bigg].
\end{align}
Noting that $\langle\psi_{S'}\otimes \varphi_{S'}|\matr{P}_Ie^{i\matr{H}t}|m\rangle=\big(\langle m|e^{-i\matr{H}t}\matr{P}_I|\psi_{S'}\otimes \varphi_{S'}\rangle\big)^{\dagger}$, we have
\begin{equation}
\mathbb{E}\big[N_{S'}^\up(t)\big]=\sum_{m}\lambda_m\,\mathbb{E}\Big[\big\|\langle m|e^{-i\matr{H}t}\matr{P}_I|\psi_{S'}\otimes \varphi_{S'}\rangle\big\|_F^2\Big],
\end{equation}
with $\|\matr{A}\|^2_F=\mathrm{Tr}[\matr{A}\matr{A}^\dagger]$ the Frobenius norm. Since for any matrix $\matr{A}\in\mathbb{C}^{m\times n}$ we have $\|\matr{A}\|_2\leq \|\matr{A}\|_F$, it follows that
\begin{align}
&\mathbb{E}\big[N_{S'}^\up(t)\big]\geq\sum_{m}\lambda_m\,\mathbb{E}\Big[\big\|\langle m|e^{-i\matr{H}t}\matr{P}_I|\psi_{S'}\otimes \varphi_{S'}\rangle\big\|^2\Big],
\shortintertext{and using Jensen's inequality}
&\mathbb{E}\big[N_{S'}^\up(t)\big]\geq\sum_{m}\lambda_m\,\mathbb{E}\Big[\big\|\langle m|e^{-i\matr{H}t}\matr{P}_I|\psi_{S'}\otimes \varphi_{S'}\rangle\big\|\Big]^2.\label{BoundArbitraryPsi}
\end{align}
Now the simplest initial state to consider for $S'$ is $|\psi_{S'}\rangle=|\alpha_n\rangle$, where $S'$ is in a single configuration $\alpha$ with exactly $n$ up-spins. Then, expanding $|m\rangle$ on the configuration basis we obtain
\begin{align}
\mathbb{E}\big[N_{S'}^\up(t)\big]&\geq\sum_{m}\lambda_m\,\sum_{\alpha_m}\mathbb{E}\Big[\big\|\langle \alpha_m|\matr{P}_Ie^{-i\matr{H}t}|\alpha_n\rangle\big\|\Big]^2.
\end{align}
In the localised phase, we know that for any two configurations $\alpha$, $\omega$ and at any time $\mathbb{E}\big[\|\langle \omega|\matr{P}_Ie^{-i\matr{H}t}|\alpha\rangle\|\big]\leq Ce^{-d(\alpha, \,\omega)/\zeta_0}$ for some positive $\zeta_0$. 
To progress, we assume that there exists a localisation length $\zeta$ and time $T$ such that for any $t>T$, the system saturates the localisation bounds on average, i.e. for any two configurations $\alpha$, $\omega$ and $t>T$, we have $\mathbb{E}\big[\|\langle \omega|\matr{P}_Ie^{-i\matr{H}t}|\alpha\rangle\|\big]= Ce^{-d(\alpha, \,\omega)/\zeta}$. 
Then we have
\begin{align}
\mathbb{E}\Big[\big\|\langle \alpha_{m}|\matr{P}_Ie^{-i\matr{H}t}|\alpha_{n}\rangle\big\|\Big]^2= C^2e^{-2d(\alpha_m,\,\alpha_{n})/\zeta}.
\end{align}
It follows that 
\begin{align}
\mathbb{E}\big[N_{S'}^\up(t)\big]&\geq C^2\sum_{m}\lambda_m\,\sum_{\alpha_m}e^{-2d(\alpha_m,\,\alpha_{n})/\zeta}.
\end{align}
Now let $\mathcal{N}(|S'|,n,m,d)$ be the number of configurations with $m$ up-spins located at distance $d$ from any configuration with $n$ up-spins on the $|S'|$-hypercube. A combinatorial analysis of the hypercube shows that this is
\begin{align}\label{DefofN}
&\mathcal{N}(|S'|,n,m,d)=\begin{cases}\binom{n}{d/2+(n-m)/2}\binom{|S'|-n}{d/2-(n-m)/2},& d+n-m~\textrm{even},\\
0,&\textrm{otherwise.}
\end{cases}
\end{align}
This leads to the following bounds on the disorder-averaged fraction of up-spins
\begin{align}
\mathbb{E}\big[N_{S'}^\up(t)\big]\geq n^\up(\zeta)&=C^2\sum_{m}\lambda_m\sum_{d=0}^{|S'|}\mathcal{N}(|S'|,n,m,d) e^{-2d/\zeta},\nonumber\\
&:=\mathcal{F}(|S'|,n),
\end{align}
where we defined $\mathcal{F}(|S'|,n)$ for latter convenience.
We could not find an elegant form for this quantity, although it is straightforward to calculate for given values of $|S'|$. In the strong localisation regime $\zeta\to 0$ we remark that 
\begin{equation}
\lim_{\zeta \to 0}n^\up(\zeta)=N^\up_{S'}(0)/|S'|,\label{LimiZet0}
\end{equation}
so that $\lim_{\zeta \to 0}\mathbb{E}\big[N_{S'}^\up(t)\big]=N_{S'}^\up(0)$, as expected of a strongly localised system.
The bound on $\mathbb{E}\big[N_{S'}^\up(t)\big]$ gives a lower bound on the disorder-averaged magnetization $M(t)$ since $M(t)=\mathbb{E}\big[N_{S'}^\up(t)\big]-1$. Now up-spins and down-spins can be exchanged in the entire analysis presented here. Exchanging up-spins with down-spins thus leads to a lower bound on the expected number of down-spins $\mathbb{E}\big[N_{S'}^\down(t)\big]$, which is 
\begin{align}
\mathbb{E}\big[N_{S'}^\down(t)\big]&\geq n^\down(\zeta)=\mathcal{F}(|S'|,\,|S'|-n).
\end{align} 
Thus, the magnetization is bounded by
\begin{equation}
1-n^\down(\zeta)\leq M(t)\leq n^\up(\zeta)-1,
\end{equation}
and in particular, on using Eq.~(\ref{LimiZet0}), we have $\lim_{\zeta \to 0}M(t)=M(0)$.

The above demonstration extends straightforwardly 
to arbitrary incoherent superposition of configurations available to $S'$, i.e.  $\rho_{\mathbb{S}}=\sum_{\alpha}|c_\alpha|^2 |\alpha\rangle\langle \alpha|\otimes|\varphi_S\rangle\langle\varphi_S|$, $c_\alpha\in\mathbb{C}$ and $\sum_{\alpha}|c_\alpha|^2=1$. As per Eq.~(\ref{BoundArbitraryPsi}), the disorder-averaged fraction of up-spins is thus bounded by
\begin{align}
\mathbb{E}\big[N_{S'}^\up(t)\big]&\geq\sum_{\alpha}|c_\alpha|^2\sum_{m}\lambda_m\,\sum_{\alpha_m}\mathbb{E}\Big[\big\|\langle \alpha_m|\matr{P}_Ie^{-i\matr{H}t}|\alpha\rangle\big\|\Big]^2,
\end{align}
where $\alpha_m$ is any configuration of $S'$ with $m$ up-spins. Using our previous results for all configurations $\alpha_n$ with exactly $n$ up-spins gives
\begin{align}
\sum_{m}\lambda_m\,\sum_{\alpha_m}\mathbb{E}\Big[\big\|\langle \alpha_m|\matr{P}_Ie^{-i\matr{H}t}|\alpha_n\rangle\big\|\Big]^2&\geq\mathcal{F}(|S'|,n).
\end{align}
Therefore the disorder-averaged fraction of up-spins is bounded as follows
\begin{align}
\mathbb{E}\big[N_{S'}^\up(t)\big]&\geq n^\up(\zeta)= \sum_{\alpha}|c_\alpha|^2\mathcal{F}\big(|S'|,\,|\alpha|\big),
\end{align}
where $|\alpha|$ is the number of up-spins in $S'$ when it is in configuration $\alpha$.
Similarly, the disorder-averaged fraction of down-spins is bounded by
\begin{align}
\mathbb{E}\big[N_{S'}^\down(t)\big]&\geq n^\down(\zeta)= \sum_{\alpha}|c_\alpha|^2\mathcal{F}\big(|S'|,\,|S'|-|\alpha|\big),
\end{align}
and $1-n^\down(\zeta)\leq M(t)\leq n^\up(\zeta)-1$. Most importantly, using Eq.~(\ref{LimiZet0}), we obtain
\begin{equation}
\lim_{\zeta\to 0}n^\up(\zeta)=\sum_{\alpha}|c_\alpha|^2\frac{|\alpha|}{|S'|}=N_{S'}^\up(0),
\end{equation}
and similarly for $\mathbb{E}\big[N_{S'}^\down(t)\big]$. These results thus show that for any $\epsilon\geq0$, there exists a localisation length $\zeta_0$ such that for all $\zeta<\zeta_0$, the magnetisation $M_t$ of $S'$ is closer to its initial value than $\epsilon$, $|M(t)-M(0)|\leq \epsilon$.

%

\section{Localisation of sublattice correlations}\label{Localcorrel}
In this section we calculate the realisation average of the two time correlation function for any two sites $i$ and $j$ of $S'$ with $I$ a bounded interval of energy
\begin{equation}
\tau_{i,j}(t)=\mathbb{E}\Big[\mathrm{Tr}[\matr{P}_I\sigma_z^i(t)\matr{P}_I\sigma_z^j \rho_{\mathbb{S}}]\Big],
\end{equation}
where $\rho_{\mathbb{S}}$ is the density matrix of the system. For the sake of clarity, from now on we let $|\omega\rangle$ denote $|\omega\rangle\otimes \matr{I}_S$ and similarly for $|\alpha\rangle$.
First, we remark that 
\begin{align}
&\tau_{i,j}(t)=\mathbb{E}\Big[\mathrm{Tr}\big[\matr{P}_Ie^{i\matr{H}t}(\matr{I}^{i}+\sigma_z^i)e^{-i\matr{H}t}\matr{P}_I\sigma_z^j \rho_{\mathbb{S}}\big]\Big]-\mathrm{Tr}[\sigma_{z}^j\rho_{\mathbb{S}}],\label{EqTAU}
\end{align}
with $\matr{I}^{i}$ the identity matrix on site $i$ and $\mathrm{Tr}[\sigma_{z}^j\rho_{\mathbb{S}}]$ is the initial expectation of $\sigma_{z}^j$, which is known. Now we consider the situation where $\rho_{\mathbb{S}}=|\alpha\rangle\langle\alpha|\otimes |\varphi_S\rangle\langle\varphi_S|$. Then
\begin{align}
&\mathbb{E}\Big[\mathrm{Tr}\big[\matr{P}_Ie^{i\matr{H}t}(\matr{I}^{i}+\sigma_z^i)e^{-i\matr{H}t}\matr{P}_I\sigma_z^j \rho_{\mathbb{S}}\big]\Big]\nonumber\\
&\hspace{5mm}=2\,(-1)^{\delta_{\down_j,\alpha}}\sum_{\omega:\, \uparrow_{i}}\mathbb{E}\Big[\mathrm{Tr}\big[\langle\alpha|\matr{P}_Ie^{i\matr{H}t}|\omega\rangle\langle\omega |e^{-i\matr{H}t}\matr{P}_I|\alpha\rangle\big]\Big],\nonumber\\
&\hspace{5mm}=2\,(-1)^{\delta_{\down_j,\alpha}}\sum_{\omega:\, \uparrow_{i}}\mathbb{E}\Big[\|\langle\omega |e^{-i\matr{H}t}\matr{P}_I|\alpha\rangle\|_F^2\Big],
\end{align}
where the sum runs over configurations $\omega$ with an up-spin at site $i$ and $\delta_{\down_j,\,\alpha}=1$ if $j$ is down in $\alpha$ and $0$ otherwise. 
Then, $\|\matr{A}\|_2\leq \|\matr{A}\|_F$ together with Jensen's inequality gives
\begin{equation}
\sum_{\omega:\, \uparrow_{i}}\mathbb{E}\Big[\|\langle\omega |e^{-i\matr{H}t}\matr{P}_I|\alpha\rangle\|_F^2\Big]\geq
\sum_{\omega:\, \uparrow_{i}}\mathbb{E}\Big[\|\langle\omega |e^{-i\matr{H}t}\matr{P}_I|\alpha\rangle\|\Big]^2.
\end{equation}
Now, in the localised phase, we know that for any two configurations $\alpha$, $\omega$ and at any time $\mathbb{E}\big[\|\langle \omega|\matr{P}_Ie^{-i\matr{H}t}|\alpha\rangle\|\big]\leq Ce^{-d(\alpha, \,\omega)/\zeta_0}$ for some positive $\zeta_0$. To progress, we assume that there exists a localisation length $\zeta$ and time $T$ such that for any $t>T$, the system saturates the localisation bounds on average, i.e. for any two configurations $\alpha$, $\omega$ and $t>T$, we have $\mathbb{E}\big[\|\langle \omega|\matr{P}_Ie^{-i\matr{H}t}|\alpha\rangle\|\big]= Ce^{-d(\alpha, \,\omega)/\zeta}$. Then we have
\begin{align}
\sum_{\omega:\, \uparrow_{i}}\mathbb{E}\Big[\|\langle\omega |e^{-i\matr{H}t}\matr{P}_I|\alpha\rangle\|_F^2\Big]&\geq\sum_{\omega:\, \uparrow_{i}}C^2e^{-2d(\alpha,\omega)/\zeta}:=\mathcal{K}(|S'|,\,|\alpha|),
\end{align}
where we have introduced $\mathcal{K}$ for later convenience. This quantity is explicitly given by
\begin{align}
&\mathcal{K}(|S'|,\,|\alpha|)=C^2\delta_{\up_i,\,\alpha}\hspace{-1.3mm}\sum_{m=1}^{|S'|-1}\sum_{d=0}^{|S'|-1}\hspace{-1mm}\mathcal{N}(|S'|-1,|\alpha|-1,m-1,d) e^{-2d/\zeta}\nonumber\\
&\hspace{5mm}+C^2(1-\delta_{\up_i,\,\alpha})\sum_{m=1}^{|S'|-1}\hspace{-.3mm}\sum_{d=1}^{|S'|-1}\hspace{-1mm}\mathcal{N}(|S'|-1,|\alpha|,m-1,d-1) e^{-2d/\zeta},\label{StepStepNoEnd}
\end{align}
with $\delta_{\up_i,\,\alpha}=1$ if $\alpha$ has an up-spin at site $i$ and 0 otherwise and $|\alpha|$ is the number of up-spins in $S'$ when it is in configurations $\alpha$.
We obtain this expression using the results of the previous section. Indeed, recall that $\mathcal{N}(|S'|,n,m,d)$ is the number of configurations with $m$ up-spins at distance $d$ from a configuration with $n$ up-spins on the $|S'|$-hypercube, see Eq.~(\ref{DefofN}). Then the number of configurations with $m$ up-spins that have an up-spin at site $i$ at distance $d$ from a configuration with $n$ up-spins with an up-spin at site $i$ is $\mathcal{N}(|S'|-1,n-1,m-1,d)$. Similarly, the number of configurations with $m$ up-spins that have an up-spin at site $i$ at distance $d$ from a configuration with $n$ up-spins with a down-spin at site $i$ is $\mathcal{N}(|S'|-1,n,m-1,d-1)$. We could not obtain an elegant form for the quantity $\mathcal{K}(|S'|,\,|\alpha|)$ although it is straightforward to evaluate for a given value of $|S'|$ and fixed configuration $\alpha$.
Finally, since $\mathcal{K}(|S'|,\,|\alpha|)\geq0$ and for small $\zeta$, $\mathcal{K}(|S'|,\,|\alpha|)\leq 1$, inserting Eq.~(\ref{StepStepNoEnd}) into Eq.~(\ref{EqTAU}) gives
\begin{equation}\label{BoundTau1}
\begin{cases}
\tau_{i,j}(t)\leq 2\mathcal{K}(|S'|,\,|\alpha|)-1,&\text{if spin at site $j$ is up in }\alpha,\\
\tau_{i,j}(t)\geq -2\mathcal{K}(|S'|,\,|\alpha|)+1,&\text{otherwise}.
\end{cases}
\end{equation}
These bounds can be completed by another set of bounds obtained on repeating the procedure presented above but starting with 
\begin{align}
&\tau_{i,j}(t)=\mathrm{Tr}[\sigma_{z}^j\rho_{\mathbb{S}}]-\mathbb{E}\Big[\mathrm{Tr}\big[\matr{P}_Ie^{i\matr{H}t}(\matr{I}^{i}-\sigma_z^i)e^{-i\matr{H}t}\matr{P}_I\sigma_z^j \rho_{\mathbb{S}}\big]\Big].
\end{align}
We obtain
\begin{equation}\label{BoundTau2}
\begin{cases}
\tau_{i,j}(t)\geq -2\mathcal{Q}(|S'|,\,|\alpha|)+1,&\text{if spin at site $j$ is up in }\alpha,\\
\tau_{i,j}(t)\leq 2\mathcal{Q}(|S'|,\,|\alpha|)-1,&\text{otherwise},
\end{cases}
\end{equation}
with 
\begin{align}
\mathcal{Q}(|S'|,\,|\alpha|)=\,&C^2(1\hspace{-.3mm}-\hspace{-.3mm}\delta_{\up_i,\,\alpha})\hspace{-1mm}\sum_{m=1}^{|S'|-1}\sum_{d=0}^{|S'|-1}\hspace{-1mm}\mathcal{N}(|S'|\hspace{-.3mm}-\hspace{-.3mm}1,|\alpha|\hspace{-.3mm}-\hspace{-.3mm}1,m\hspace{-.3mm}-\hspace{-.3mm}1,d) e^{-2d/\zeta}\nonumber\\
&+C^2\delta_{\up_i,\,\alpha}\sum_{m=1}^{|S'|-1}\sum_{d=1}^{|S'|-1}\mathcal{N}(|S'|-1,|\alpha|,m-1,d-1) e^{-2d/\zeta}.
 \end{align}
Taken together Eqs.~(\ref{BoundTau1}) and (\ref{BoundTau2}) guarantee that once our criterion for many-body dynamical localisation holds for a sublattice comprising $|S'|$ sites, then the correlation between any two spins at sites of any such sublattice is constrained in an interval centered on its initial value.

\bibliographystyle{plain}
\bibliography{Anderson_V7}

\end{document}